\pdfoutput=1
\documentclass[11pt, letterpaper]{article}
\usepackage{blindtext}
\usepackage{hyperref}
\hypersetup{
	colorlinks=true,
	linkcolor=blue!70!black,
	citecolor=blue!70!black,
	urlcolor=blue!70!black
}

\usepackage[english]{babel}
\usepackage[T1]{fontenc}
\usepackage{thm-restate}

\usepackage[margin=1in]{geometry}

\usepackage{amsmath}
\usepackage{amssymb}
\usepackage{amsthm}
\usepackage{booktabs}
\usepackage{algorithm}
\usepackage[lined,boxed,ruled,norelsize,algo2e,linesnumbered]{algorithm2e}
\usepackage{bbold}
\usepackage{bbm}

\usepackage{graphicx}
\graphicspath{{./figs/}}

\usepackage{cleveref}
\newtheorem{theorem}{Theorem}
\newtheorem{lemma}{Lemma}
\newtheorem{fact}{Fact}
\newtheorem{definition}{Definition}
\newtheorem{corollary}{Corollary}
\newtheorem{proposition}{Proposition}
\newtheorem{claim}{Claim}

\newtheorem{assumption}{Assumption}
\newtheorem{remark}{Remark}

\newcommand{\defeq}{:=}
\newcommand{\norm}[1]{\left\lVert#1\right\rVert}
\newcommand{\norms}[1]{\lVert#1\rVert}
\newcommand{\normf}[1]{\left\lVert#1\right\rVert_{\textup{F}}}

\newcommand{\normop}[1]{\left\lVert#1\right\rVert_{\textup{op}}}
\newcommand{\normsop}[1]{\lVert#1\rVert_{\textup{op}}}
\newcommand{\normtr}[1]{\left\lVert#1\right\rVert_{\textup{tr}}}

\newcommand{\inprod}[2]{\left\langle#1, #2\right\rangle}

\newcommand{\eps}{\varepsilon}
\newcommand{\lam}{\lambda}

\newcommand{\R}{\mathbb{R}}

\newcommand{\N}{\mathbb{N}}
\newcommand{\diag}[1]{\textbf{\textup{diag}}\left(#1\right)}
\newcommand{\half}{\frac{1}{2}}

\newcommand{\1}{\mathbbm{1}}
\newcommand{\E}{\mathbb{E}}
\newcommand{\ME}{\mathcal{E}}

\newcommand{\Vol}{\textup{Vol}}
\newcommand{\Nor}{\mathcal{N}}

\newcommand{\Tr}{\textup{Tr}}
\newcommand{\opt}{\textup{OPT}}
\newcommand{\xset}{\mathcal{X}}

\newcommand{\ma}{\mathbf{A}}

\newcommand{\id}{\mathbf{I}}

\usepackage{xcolor}
\definecolor{burntorange}{rgb}{0.8, 0.33, 0.0}

\newcommand{\tO}{\widetilde{O}}
\newcommand{\nnz}{\textup{nnz}}

\newcommand{\Par}[1]{\left(#1\right)}
\newcommand{\Brack}[1]{\left[#1\right]}
\newcommand{\Brace}[1]{\left\{#1\right\}}
\newcommand{\Abs}[1]{\left|#1\right|}

\newcommand{\oracle}{\mathcal{O}}

\newcommand{\Sym}{\mathbb{S}}
\newcommand{\PSD}{\Sym_{\succeq \mzero}}
\newcommand{\PD}{\Sym_{\succ \mzero}}
\newcommand{\mzero}{\mathbf{0}}
\newcommand{\alla}{\mathcal{A}}
\newcommand{\talla}{\widetilde{\mathcal{A}}}
\newcommand{\my}{\mathbf{Y}}
\newcommand{\mb}{\mathbf{B}}
\newcommand{\0}{\mathbb{0}}
\newcommand{\msig}{\boldsymbol{\Sigma}}
\newcommand{\set}{\mathcal{K}}
\newcommand{\ctight}{c_{\textup{tight}}}
\newcommand{\Cset}{C_{\textup{set}}}
\newcommand{\rs}{r_\star}
\newcommand{\mm}{\mathbf{M}}
\newcommand{\allm}{\mathcal{M}}
\newcommand{\xs}{x_\star}
\newcommand{\tma}{\widetilde{\ma}}
\newcommand{\tx}{\tilde{x}}

\newcommand{\tmv}{\mathcal{T}_{\textup{mv}}}
\newcommand{\mx}{\mathbf{X}}
\newcommand{\blam}{\bar{\lam}}
\newcommand{\lamt}{\lam_{\textup{test}}}
\newcommand{\xt}{x_{\textup{test}}}
\newcommand{\At}{A_{\textup{test}}}
\newcommand{\bmy}{\overline{\my}}
\newcommand{\csparse}{C_{\textup{sparse}}}
\newcommand{\bw}{\bar{w}}
\newcommand{\bv}{\bar{v}}
\newcommand{\mw}{\mathbf{W}}
\newcommand{\ms}{\mathbf{S}}

\newcommand{\lap}{\mathbf{L}}

\newcommand{\idnoone}{\id_{V \setminus \1}}
\newcommand{\mc}{\mathbf{C}}
\newcommand{\rsc}{\rs^{\mc}}
\newcommand{\xsc}{\xs^{\mc}}
\newcommand{\xsetc}{\xset^{\mc}}
\newcommand{\olinopt}{\oracle_{\textup{lin-opt}}}

\newcommand{\mn}{\mathbf{N}}
\newcommand{\ball}{\mathbb{B}}
\newcommand{\tms}{\widetilde{\ms}}
\newcommand{\tmb}{\widetilde{\mb}}
\newcommand{\dd}{\textup{d}}
\newcommand{\tlap}{\widetilde{\lap}}
\newcommand{\proj}{\boldsymbol{\Pi}}
\newcommand{\mr}{\mathbf{R}}
\newcommand{\troute}{\mathcal{T}_{\textup{route}}}
\newcommand{\up}{\textup{up}}
\newcommand{\lo}{\textup{lo}}
\newcommand{\md}{\mathbf{D}}
\newcommand{\tg}{\tilde{g}}
\newcommand{\bO}{\Breve{O}}
\renewcommand{\O}{\widetilde{O}} 
\title{Linear-Sized Sparsifiers via \\ Near-Linear Time Discrepancy Theory}
\author{Arun Jambulapati\thanks{University of Washington, {\tt jmblpati@uw.edu}} \and Victor Reis\thanks{University of Washington, {\tt voreis@cs.washington.edu}} \and Kevin Tian\thanks{Microsoft Research, {\tt tiankevin@microsoft.com}}}
\date{}

\begin{document}

\maketitle

\begin{abstract}
Discrepancy theory has provided powerful tools for producing higher-quality objects which ``beat the union bound'' in fundamental settings throughout combinatorics and computer science. However, this quality has often come at the price of more computationally-expensive algorithms. We introduce a new framework for bridging this gap, by allowing for the efficient implementation of discrepancy-theoretic primitives. Our framework repeatedly solves regularized optimization problems to low accuracy to approximate the partial coloring method of \cite{Rothvoss14}, and simplifies and generalizes recent work of \cite{JainSS23} on fast algorithms for Spencer's theorem. In particular, our framework only requires that the discrepancy body of interest has exponentially large Gaussian measure and is  expressible as a sublevel set of a symmetric, convex function. We combine this framework with new tools for proving Gaussian measure lower bounds to give improved algorithms for a variety of sparsification and coloring problems.

As a first application, we use our framework to obtain an $\widetilde{O}(m \cdot \epsilon^{-3.5})$ time algorithm for constructing an $\epsilon$-approximate spectral sparsifier of an $m$-edge graph, matching the sparsity of \cite{BatsonSS14} up to constant factors and improving upon the $\widetilde{O}(m \cdot \epsilon^{-6.5})$ runtime of \cite{LeeS17}. We further give a state-of-the-art algorithm for constructing graph ultrasparsifiers and an almost-linear time algorithm for constructing linear-sized degree-preserving sparsifiers via discrepancy theory; in the latter case, such sparsifiers were not known to exist previously. We generalize these results to their analogs in sparsifying isotropic sums of positive semidefinite matrices. Finally, to demonstrate the versatility of our technique, we obtain a nearly-input-sparsity time constructive algorithm for Spencer's theorem (where we recover a recent result of \cite{JainSS23}).
\end{abstract}

\newpage

\pagenumbering{gobble}
\setcounter{tocdepth}{2}
{
\tableofcontents
}
\newpage
\pagenumbering{arabic}

\section{Introduction}\label{sec:intro}

Throughout the history of theoretical computer science and combinatorics, the development of discrepancy theory has yielded techniques for producing high-quality objects which minimize deviation from typical behavior. As a well-known example which has received significant recent attention, Spencer's ``six standard deviations suffice'' theorem \cite{Spencer85} says that for any matrix $\ma \in \{0, 1\}^{n \times n}$, whose rows indicate $n$ sets $S_i \subseteq [n]$, there exists at least one coloring $x \in \{-1, 1\}^n$ such that $\norm{\ma x}_\infty \le 6\sqrt{n}$; we remark that the problem has a combinatorial interpretation as choosing colors to minimize the largest color discrepancy in any set. Notably, this result improves upon the quality of a random coloring $x$, which achieves $\norm{\ma x}_\infty = \Theta(\sqrt{n \log n})$ with high probability for random set systems, by standard concentration and anti-concentration arguments. Discrepancy-theoretic arguments of a similar nature have led to significant progress in a variety of fields spanning computational geometry, complexity theory, approximation algorithms, numerical analysis, graph theory, and experimental design \cite{Matousek99, Chazelle01, BatsonSS14, HarshawSSZ19, ZadehBGNSS22}. 

However, the development of computationally-efficient algorithms for discrepancy minimization is comparatively rather nascent. To use Spencer's theorem as an example, while \cite{Spencer85} shows the existence of low-discrepancy colorings, these colorings are highly atypical and it is not immediately clear how to find one algorithmically. Roughly a decade ago, a sequence of works \cite{Bansal10, LovettM15, HarveySS14, Rothvoss14, EldanS18} gave \emph{constructive discrepancy} algorithms in settings including Spencer's set coloring problem and variants thereof, settling the polynomial-time computability of high-quality colorings of set systems. The works of \cite{Rothvoss14, EldanS18} in particular provided very general frameworks for discrepancy minimization. For example, given a symmetric, convex ``discrepancy body'' $\set$ (e.g.\ $\set \defeq \{x \in \R^n \mid \norm{\ma x}_\infty = O(\sqrt n)\}$ in Spencer's setting) with an $\exp(-O(n))$ Gaussian measure lower bound, \cite{Rothvoss14} showed that the nearest point in $\set \cap [-1, 1]^n$ to a random Gaussian vector has a constant fraction of tight hypercube constraints with high probability. A similar statement was shown by \cite{EldanS18} for the solution to a suitable linear program, and these ``partial coloring by Gaussian rounding'' subroutines naturally induce full discrepancy minimization algorithms via recursion.

While the constructive frameworks of \cite{Bansal10, LovettM15, HarveySS14, Rothvoss14, EldanS18} are elegant and simple to state, their straightforward implementation requires the black-box use of powerful convex programming primitives such as high-accuracy linear or semidefinite programming solvers, which incur substantial runtime overhead. This has led to a recurring theme in algorithmic discrepancy theory: attaining higher object quality may come at the price of worse computational efficiency.

A recent exciting work by \cite{JainSS23} was able to overcome this hurdle for the specific setting of Spencer's theorem, by leveraging \emph{approximate} linear programming solvers running in near-linear time. Specifically, \cite{JainSS23} designed an algorithm which produces a coloring $x \in \{-1, 1\}^n$ satisfying $\norm{\ma x}_\infty = O(\sqrt n)$ in time $\tO(\nnz(\ma) + n)$, by efficiently solving the linear programs required by the \cite{EldanS18} framework and quantifying their approximation tolerance. However, the \cite{JainSS23} analysis relied on properties tailored to the setting of Spencer's theorem, specifically an exponential lower bound on the probability that a random coloring achieves $\norm{\ma x}_\infty = O(\sqrt n)$. Such a statement is not known to hold in various other discrepancy-theoretic settings of interest, such as the Koml\'os problem, even when the corresponding Gaussian measure lower bound required by the \cite{Rothvoss14, EldanS18} frameworks is known. Our work is motivated by the current lack of a general-purpose, approximation-tolerant framework for discrepancy minimization under the minimal requirements of symmetry, convexity, and a Gaussian measure lower bound, which may pave the way towards near-linear time algorithms in discrepancy-theoretic settings beyond Spencer's theorem.

\subsection{Our results}\label{ssec:results}

\paragraph{Discrepancy minimization framework.} We introduce a new \emph{approximation-tolerant} variant of the Gaussian rounding framework of \cite{Rothvoss14}. Before stating its guarantees, we first briefly recall the main result of \cite{Rothvoss14}, deferring a more extended discussion to Section~\ref{sec:framework}. Given a symmetric convex set\footnote{Throughout the exposition in this section, we let the dimensionality of our coloring variable $x$ be denoted $m$ for consistency with our graph-theoretic applications, and state the implications of our results in Spencer's setting (whose dimension is typically denoted by $n$) consistently in an abuse of notation.} $\set \subseteq \R^m$ with $\gamma_m(\set) = \exp(-O(m))$ (see Section~\ref{sec:prelims} for notation), Theorem 2 of \cite{Rothvoss14} shows that with high probability over a randomly sampled Gaussian vector $g \sim \Nor(\0_m, \id_m)$, letting 
\[\xs \defeq \arg\min_{x \in \set \cap [-1, 1]^m} \norm{x - g}_2,\]
we have $|\{i \mid |[\xs]_i| = 1\}| = \Omega(m)$. In other words, by rounding $g$ to the intersection of $\set$ and $[-1, 1]^m$, a constant fraction of the hypercube constraints $|[\xs]_i| \le 1$ are saturated. When $\set$ is taken to be a discrepancy body, e.g.\ $\set = \{x \in \R^m \mid \norm{\ma x}_\infty = O(\sqrt m)\}$ in Spencer's setting, the result of \cite{Rothvoss14} implies that we can freeze a constant fraction of colors at $\pm 1$ without incurring much discrepancy, and recurse on the uncolored coordinates. High-precision computation of $\xs$ is further polynomial-time tractable, as it is the solution to a convex optimization problem.

Our approximation-tolerant framework also applies to any symmetric, convex $\set \subseteq \R^m$ with large $\gamma_m(\set)$ (as in \cite{Rothvoss14}), but is most easily described when
\[\set = \Brace{x \in \R^m \mid f(x) \le \rho},\]
where $\set$ is expressed as a sublevel set of a symmetric, convex function $f: \R^m \to \R$, and $\rho \ge 0$ is a ``discrepancy radius.'' All symmetric convex sets trivially satisfy this description by letting $\rho = \infty$ and $f$ be the indicator function of the set, but in many applications of interest (explored throughout the paper), $\set$ is naturally characterized as a sublevel set with finite $\rho$. In Section~\ref{sec:framework}, we show that for a parameter $\beta \in (0, 1)$, solving $\tO(1)$ regularized problems of the form\footnote{Here and throughout the paper, we use the notation $\O(\cdot)$ to hide polynomial factors in $\log n$ and $\log(\frac{1}{\eps})$ and $\bO(\cdot)$ to hide polynomial factors in $\log \log n$ and $\log\log(\frac{1}{\eps})$. We defer our notational conventions to Section~\ref{sec:prelims}.}
\[\min_{x \in [-1, 1]^m} f(x) + \lam\norm{x - g}_2^2\]
to additive accuracy $\Theta(\rho\beta^2)$ suffices to produce a point $x \in \set \cap [-1, 1]^m$ with $|\{i \mid |x_i| \ge 1 - \beta\}| = \Omega(m)$. In other words, by approximately solving a sequence of regularized optimization problems to accuracy $\Theta(\rho\beta^2)$, our framework yields a point with low discrepancy and many nearly-saturated hypercube constraints. In applications such as Spencer's setting, this additive accuracy is fairly generous: taking $\beta$ to be inverse-polylogarithmic means we only need to optimize the regularized objective to error $\approx \sqrt m$, which is efficiently implementable in input-sparsity time (after applying reductions from \cite{JainSS23}) via stochastic first-order methods. Furthermore, these nearly-saturated coordinates of our approximate solution $x$ can then be randomly rounded to exactly be in $\{\pm 1\}$ without incurring much additional discrepancy. In our graph-theoretic applications of our approximation-tolerant framework, we handle the nearly-saturated coordinates recursively.

As a direct application of our framework in Section~\ref{sec:framework}, we show how to combine it with existing stochastic optimization methods from \cite{CarmonJST20} to recover the main result of \cite{JainSS23} through arguably a simpler approach, which does not rely on structural properties of Spencer's problem beyond a Gaussian measure lower bound. This result is provided in Section~\ref{sec:spencer}.

\paragraph{Linear-sized sparsifiers.} We next illustrate how our framework allows us to leverage powerful discrepancy-theoretic tools to design fast algorithms for the problem of graph sparsification. Given a parameter $\eps \in (0, 1)$ and a graph $G$ on $n$ vertices and $m$ edges, with Laplacian $\lap_G \in \R^{n \times n}$, the goal is to find a reweighted subgraph $H$ of $G$ with much fewer edges and Laplacian $\lap_H$, satisfying
\[(1 - \eps)\lap_G \preceq \lap_H \preceq (1 + \eps)\lap_G.\]
A well-known result of \cite{SpielmanS11} shows that randomly sampling edges of $G$ proportionally to their effective resistances yields a reweighted subgraph $H$ satisfying the above guarantee with $O(n\log n \cdot \eps^{-2})$ edges. The proof of correctness of such a sampling strategy is via a direct application of a matrix Chernoff concentration bound. Further, by using sketching techniques and specialized linear system solvers, the \cite{SpielmanS11} sparsifiers can be constructed in nearly-linear time $\tO(m)$. However, this size bound is not optimal: a breakthrough work by \cite{BatsonSS14} gave a polynomial-time algorithm for constructing a spectral sparsifier $H$ with $O(n \cdot \eps^{-2})$ edges, later shown to be tight \cite{CarlsonKST19}. We call such a subgraph $H$ a linear-sized sparsifier (an accurate description for constant $\eps$).

The \cite{BatsonSS14} sparsification result is discrepancy-theoretic in nature, as it improves the quality of a randomly sampled sparsifier (where the quality is the number of edges). Interestingly, however, the techniques used to construct linear-sized sparsifiers are quite different than those employed by the aforementioned works \cite{Bansal10, LovettM15, HarveySS14, Rothvoss14, EldanS18}. Indeed, the original announcement of \cite{BatsonSS14} predates all of these works, so at the time Spencer's theorem was not even algorithmic. Instead, the sparsifier of \cite{BatsonSS14} is constructed one edge at a time, guided by the careful use of a potential function which controls the discrepancy. A sequence of follow-up works \cite{ZhuLO15, LeeS18} culminated in the current state-of-the-art algorithm of \cite{LeeS17}, which obtains a linear-sized sparsifier in time $\tO(m \cdot \eps^{-6.5})$, using structured semidefinite program solvers to efficiently implement updates against a modification of \cite{BatsonSS14}'s potential based on the trace of a matrix exponential.

We give an improved algorithm for linear-sized sparsification via our framework. We prove the following in Section~\ref{sec:bss} as a special case of a result on sparsifying isotropic matrix sums (Theorem~\ref{thm:fast_bss}). Our result almost quadratically improves upon the $\eps$ dependence of \cite{LeeS17}.

\begin{restatable}{theorem}{restatebssgraph}\label{thm:bss_graph}
Given a graph $G = (V, E, w_G)$ with $m = |E|$, $n = |V|$, and $\eps \in (0, 1)$, there is a randomized algorithm which in time $\bO(m\log^{4}(m) \log (\frac 1 \eps) \cdot \eps^{-3.5})$ returns $w \in \R^E_{\ge 0}$ satisfying $\nnz(w) = O(\frac n {\eps^2})$, such that with high probability in $n$,
\[(1 - \eps)\lap_G \preceq \sum_{e \in E} w_e b_e b_e^\top \preceq (1 + \eps)\lap_G.\]
\end{restatable}

We remark that Theorem~\ref{thm:bss_graph} follows from an instantiation of our framework in Section~\ref{sec:framework}, using an approximate Laplacian linear system solver of \cite{KoutisMP11} and a recent approximate semidefinite program solver from \cite{JambulapatiT23} to solve the resulting subproblems. Both of these results lie in active areas of research, and improvements therein immediately imply faster algorithms for linear-sized sparsification via our framework. We find it a promising proof-of-concept that we rely on more conventional discrepancy-theoretic techniques (motivated by \cite{Rothvoss14, RR20}) to establish Theorem~\ref{thm:bss_graph}, eschewing conventional wisdom that the use of such techniques comes at a computational price.

To our knowledge, no analog of the structural fact used in \cite{JainSS23} (where the set of sparse reweightings which spectrally approximate $G$ has inverse-exponentially large measure) is known in the setting of spectral sparsification. In demonstrating Theorem~\ref{thm:bss_graph}, it is thus important that our framework only relies on measure lower bounds known for the relevant operator norm balls.

\paragraph{Ultrasparsifiers.} Ultrasparsifiers were introduced in \cite{SpielmanT04} to dramatically reduce the number of edges in a graph at a cost of a larger approximation factor, and can be viewed as ``one-sided $\eps > 1$'' variants of the sparsifiers in our earlier discussion. We define ultrasparsifiers as follows.

\begin{definition}[Graph ultrasparsifiers]\label{def:ultra}
Let $G$ be a graph on $n$ vertices. A graph $H$ is an $(\kappa, \ell)$-ultrasparsifier of $G$ if it has $n-1+\frac{n}{\ell}$ edges and satisfies
\[
\lap_H \preceq \lap_G \preceq \kappa \lap_H.
\]
\end{definition}

In general, taking $\ell \to n$ and the fact that no tree is better than an $O(n)$ approximation of a clique implies the tightest ultrasparsifiers we can hope for in all regimes must have $\kappa = \Omega(\ell)$. The first construction of graph ultrasparsifiers was achieved by \cite{SpielmanT04}, who gave a nearly-linear time algorithm to construct $(\ell\log^{O(1)}(n), \ell)$-ultrasparsifiers for any $\ell \ge 1$. By applying a matrix Chernoff-based random sampling argument, \cite{KoutisMP11} then gave a simpler construction of $(\bO(\ell\log^{2}(n)), \ell)$-ultrasparsifiers in nearly-linear time. Further, using techniques inspired by the \cite{BatsonSS14} potential-based argument, higher-quality ultrasparsifiers are known to be constructible in polynomial time: \cite{KollaMST10} constructed $(\bO(\ell\log(n)), \ell)$-ultrasparsifiers based on state-of-the-art low-stretch spanning trees, and recently \cite{JambulapatiS21} constructed $(\ell^{1+o(1)}\log^{o(1)}(n), \ell)$-ultrasparsifiers by going beyond the low-stretch spanning tree framework. To the best of our knowledge, it is unknown how to adapt the tools developed for fast linear-sized sparsification in \cite{ZhuLO15, LeeS18, LeeS17} to the ultrasparsifier setting, limiting the algorithmic uses of the higher-quality ultrasparsifiers of \cite{KollaMST10, JambulapatiS21}.

In Section~\ref{sec:ultra}, we prove Theorem~\ref{thm:ultra}, a generalization of Theorem~\ref{thm:bss_graph} which extends our discrepancy minimization techniques to the ultrasparsifier setting. In fact, Theorem~\ref{thm:ultra} follows as a corollary of Theorem~\ref{thm:bss_graph} in light of a new discrepancy-theoretic result we prove as Theorem~\ref{thm:gen_rr20}. We then apply Theorem~\ref{thm:ultra} to obtain the following ultrasparsifier construction, matching the guarantees of \cite{KollaMST10, JambulapatiS21} in nearly-linear time and potentially paving the way for their use in fast algorithms.

\begin{restatable}{corollary}{restateultrafast}\label{cor:ultra_fast}
Given a graph $G = (V, E, w_G)$ with $m = |E|$, $n = |V|$, and $\ell \ge 1$, there is a randomized algorithm which in time $\bO(m\log^4(m))$ returns a $(\min(\bO(\ell\log(n)), \ell^{1 + o(1)}\log^{o(1)}(n)), \ell)$-ultrasparsifier of $G$, with high probability in $n$.
\end{restatable}

\paragraph{Degree-preserving sparsifiers.} As a final application of our framework, we turn our attention to degree-preserving sparsification, a primitive introduced by \cite{ChuGPSSW18} which asks for a spectral sparsifier that also preserves the weighted degrees of all vertices. Degree-preserving sparsification has arisen as a natural middle ground between undirected graph sparsification and directed graph sparsification. Concretely, Eulerian sparsifiers of Eulerian graphs have become the de facto notion of directed graph sparsification due to e.g.\ \cite{CohenKPPSV16, PengS22}, which show how to reduce linear system solving in directed graph Laplacians to Eulerian sparsification. A technique known as short-cycle decomposition, introduced in \cite{ChuGPSSW18} and refined in \cite{LiuSY19, ParterY19}, was previously used to obtain efficient constructions of both (undirected) degree-preserving sparsifiers and (directed) Eulerian sparsifiers by repeatedly performing degree-preserving operations. This motivates the study of degree-preserving sparsification as a stepping stone towards more complex notions of sparsifiers.

Perhaps surprisingly, even the existence of linear-sized degree-preserving sparsifiers (with $O(n \cdot \eps^{-2})$ edges) is not known. The state-of-the-art existential result in the literature is degree-preserving sparsifiers with $O(n\log^2 n \cdot \eps^{-2})$ edges, constructible in polynomial time \cite{ChuGPSSW18, ParterY19}; we note that \cite{ParterY19} also gives an almost-linear time construction of $O(n\log^3 n \cdot \eps^{-2})$-sized degree-preserving sparsifiers.\footnote{We additionally remark that a size bound of $O(n\log n \cdot \eps^{-2})$ is known as folklore in the community: it follows from combining the short-cycle decomposition of \cite{ChuGPSSW18} with \cite{BatsonSS14}.} In Section~\ref{sec:degree}, we give the following improved degree-preserving sparsification result.

\begin{restatable}{theorem}{restatemaindegree}\label{thm:main_degree}
Given a graph $G = (V, E, w_G)$ with $m = |E|$, $n = |V|$, and $\eps \in (0, 1)$, there is a randomized algorithm which returns $w \in \R^E_{\ge 0}$ satisfying $\nnz(w) = O(\frac n {\eps^2})$, $|\mb|^\top w = |\mb|^\top w_G$ in time
\[O\Par{(m + \troute(\alpha))\cdot \alpha^2 \cdot \textup{poly}\Par{\frac{\log m}{\eps}}},\]
for any $\alpha \ge 1$ following the notation of Definition~\ref{def:route}, such that with high probability in $n$,
\[(1 - \eps)\lap_G \preceq \sum_{e \in E} w_e b_e b_e^\top \preceq (1 + \eps)\lap_G.\]
\end{restatable}

The runtime of Theorem~\ref{thm:main_degree} is parameterized by the quality and runtime of state-of-the-art \emph{oblivous routings}, defined formally in Definition~\ref{def:route}. As we recall in Proposition~\ref{prop:klos}, current oblivous routings \cite{KelnerLOS14} achieve $\troute(\alpha) = mn^{o(1)}$ for $\alpha = n^{o(1)}$, so Theorem~\ref{thm:main_degree} runs in almost-linear time $m^{1 + o(1)} \cdot \text{poly}(\eps^{-1})$. However, polylogarithmic-quality \emph{congestion approximators} are known to be constructible in time $m\cdot\text{polylog}(n)$ \cite{Peng16}, which are closely-related to oblivous routings; any improvements on oblivous routing quality would then be reflected in Theorem~\ref{thm:main_degree} as well.

Finally, we note that expander decompositions (and the fact that electric routings are good oblivous routings for expanders) can improve the runtime of Theorem~\ref{thm:main_degree}, at the cost of a logarithmic overhead in the sparsity. We give an algorithm achieving this runtime-sparsity tradeoff as Theorem~\ref{thm:degree_polylog}.

\subsection{Technical overview}\label{ssec:techniques}

\paragraph{Gaussian rounding via regularized optimization.} Consider the setting of the \cite{Rothvoss14} framework, where symmetric, convex $\set \subseteq \R^m$ has $\exp(-O(m))$ Gaussian measure. Further, let $g \sim \Nor(\0_m, \id_m)$, and $\xs \in \set \cap [-1, 1]^m$ minimize the distance to $g$, (so $\xs$ has a constant fraction of coordinates at $\pm 1$). Our starting point is the observation that if $x \in \set \cap [-1, 1]^m$ satisfies $\norm{x - \xs}_2^2 = O(m\beta^2)$ for an appropriate constant, then it must have many coordinates in $[-1, -1+\beta] \cup [1 - \beta, 1]$; a simple proof is provided in Lemma~\ref{lem:many_nearly_tight}. Moreover, strong convexity of $\norm{x - g}_2^2$ then implies that to achieve this distance bound, it suffices to find $x \in \set \cap [-1, 1]^m$ minimizing the squared distance to $g$ up to $O(m\beta^2)$, formalized in Lemma~\ref{lem:sc_distsquare}. In settings where $\set = \{x \in \R^m \mid f(x) \le \rho\}$, we wish to find $x$ achieving both small function value (according to $f$) and distance to $g$. This motivates the consideration of regularized mixed objectives of the form
\[f_\lam(x) \defeq f(x) + \lam \norm{x - g}_2^2.\]
The bulk of Section~\ref{sec:framework} then develops a binary-search procedure which approximately minimizes $f_\lam$ over $[-1, 1]^m$ for different values of $\lam$, and aggregates these solutions to achieve the desired distance to $\xs$. We give a formal statement in Proposition~\ref{prop:binary_search_lam} quantifying the range of $\lam$ under consideration by our search, as well as the accuracy levels required by our subproblem solvers (for optimizing $f_\lam$).

\paragraph{New tools for lower bounding Gaussian measure.} To apply the framework of Section~\ref{sec:framework}, we require methods for proving Gaussian measure lower bounds for various discrepancy bodies. For Spencer's setting, since the discrepancy body is a polyhedron with $O(m)$ facets, a Gaussian measure lower bound follows from a routine application of the Sidak-Khatri correlation inequality (see e.g.\ Lemma 8.9, \cite{RothvossLectureNotes2016}). For non-polyhedral sets such as the operator norm balls arising in spectral sparsification settings, however, na\"ive applications of Sidak-Khatri fall short. This bottleneck was alleviated in part by \cite{RR20}, who proved that a different sufficient condition for the \cite{Rothvoss14} framework holds for operator norm balls $\set$ arising in spectral sparsification settings, i.e.\ $\set = \{x \in \R^m \mid \normsop{\sum_{i \in [m]} x_i \ma_i} \le \rho\}$ for matrices $\{\ma_i\}_{i \in [m]}$. Roughly speaking, the main result of \cite{RR20} shows that for any $\alpha$, with probability $\ge \half$ a random Gaussian vector lies at distance $O(\alpha \sqrt m)$ from $\frac 1 \alpha \set$ (when $\rho = \sqrt{n/m}$). This was proven by conducting a controlled random walk, where the output of the walk is close to a Gaussian and achieves good discrepancy. The same paper asked whether a direct Gaussian measure lower bound holds for the relevant operator norm balls.

In Section~\ref{sec:measure}, we augment \cite{RR20} with two new results of potentially independent interest. The first, Theorem~\ref{thm:expansionimplieslowerbound}, generically reduces Gaussian measure lower bounds for convex sets to statements of the form proven by \cite{RR20}, and thus expands the Gaussian measure toolkit for non-polyhedral sets. We use this result in conjunction with an extension of \cite{Rothvoss14} due to \cite{RR22}, which tolerates intersections with a subspace, in our degree-preserving sparsification algorithms in Section~\ref{sec:degree}.
The second, Theorem~\ref{thm:gen_rr20}, generalizes the main result of \cite{RR20} when the matrix set of interest is parameterized by an additional trace bound, and is crucial to our ultrasparsifier algorithms in Section~\ref{sec:ultra}. 

\paragraph{Spectral sparsification and box-spectraplex games.} In the remainder of the paper, we focus on applications of our framework to spectral sparsification. Our algorithms for linear-sized sparsifiers and ultrasparsifiers are essentially identical in light of Theorem~\ref{thm:gen_rr20}, so we discuss the former. We begin by observing that the subproblems required from our framework, when $\set$ is an operator-norm ball, are structured semidefinite programs. In particular, they are expressible as \emph{box-spectraplex games}, which are bilinear minimax optimization problems between a min player living in the box $[-1, 1]^m$, and a max player living in the spectraplex $\{\my \in \R^{n \times n} \mid \Tr(\my) = 1, \my \succeq \mzero_n\}$. We adapt a recent nearly-linear time approximation algorithm for box-spectraplex games from \cite{JambulapatiT23} to solve our subproblems to sufficiently high accuracy for the linear-sized sparsification recursion in \cite{RR20} to work out, up to a small error term which can be discarded. By using random initializations, we only lose $\textup{polyloglog}$ factors in runtime over the solver of \cite{JambulapatiT23}. The only additional overhead in Theorem~\ref{thm:bss_graph} is from solving a Laplacian linear system to bring the problem into isotropic position.

Executing this strategy for degree-preserving sparsification requires more care: we need to make sure all of our approximate solvers and intermediate rounding steps preserve the degree constraints. We first discuss how to solve the subproblems required by our framework. We show that a Frank-Wolfe method analyzed in the $\ell_\infty$ norm reduces the optimization problem to a small number of approximate linear optimization problems over the intersection of $[-1, 1]^m$ and circulation space (the kernel of a reweighted $\mb^\top$). We solve these linear optimization problems using oblivous routings (Definition~\ref{def:route}), which both parameterize circulation space and precondition the optimization problem. Finally, to handle the rounding error due to the recursion scheme, we develop a key subroutine implemented in nearly-linear time using dynamic data structures \cite{SleatorT83}. Our subroutine lets us round a small flow on a bipartite graph efficiently, zeroing out a constant fraction of its weights while preserving degrees, and without more than doubling any edge weight.

\subsection{Related work}\label{ssec:related}

\paragraph{Algorithmic discrepancy theory.} The theoretical computer science community has developed a number of constructive discrepancy-theoretic proofs over the past decade, in the form of polynomial-time algorithms. It is out of our scope to survey this wide body of work, but beyond those discussed earlier on linear-sized sparsification and Spencer's theorem, we refer the reader to additional examples in \cite{BansalDG19, Cohen16, BansalDGL19, DadushNTT18, BansalJ0S20, AlweissLS21, BansalLV22, PV23}. We further find it interesting to develop an approximate implementation of the framework of \cite{EldanS18} based on convex programming under minimal assumptions, analogous to our implementation of \cite{Rothvoss14}. Finally, of particular relevance to the themes of our paper are two recent algorithms for hereditary discrepancy minimization in set sytems by \cite{Larsen23, DengSW22}. The algorithm of \cite{DengSW22} runs in input-sparsity time (in the indicator matrix) in some regimes. We find it potentially fruitful to explore improving these runtimes, and broadening the range of problems to which their tools are applicable.

\paragraph{Approximate semidefinite programming.} Our sparsification algorithms in Sections~\ref{sec:bss},~\ref{sec:ultra}, and~\ref{sec:degree} all rely on approximate solvers \cite{JambulapatiT23} for structured semidefinite programs expressible as box-spectraplex games. Such problems naturally arise from our framework applied to operator norm balls, as the minimization player lives in $[-1, 1]^m$ (the box), and the dual best response lives in the spectraplex. We hence find investigating these games interesting, as any improvements to \cite{JambulapatiT23} would reflect in our runtimes, and those of any future applications. For example, could we improve runtimes via low-rank sketches which have found success in simplex-spectraplex games \cite{KalaiV05, BaesBN13, GarberHM15, Allen-ZhuL17, CarmonDST19}? Improvements beyond rank reduction are also interesting, as the matrix Spencer setting (which was nearly-resolved constructively in \cite{BansalJM23}) induces an operator norm ball, but the guarantees of current solvers lose $\textup{poly}(n)$ factors in the runtime there.

\section{Preliminaries}\label{sec:prelims}

\paragraph{General notation.} We let $\tO$ suppress polylogarithmic factors in problem parameters, and we let $\bO$ suppress polyloglogarithmic factors in problem parameters. We use $[n]$ to denote $\{i \in \N \mid i \le n\}$. The $\ell_p$ norm of a vector argument is $\norm{\cdot}_p$. The all-zeroes and all-ones vectors of dimension $d$ are $\0_d$ and $\1_d$. For $S \subseteq [d]$ where $d$ is clear from context, $\1_S \in \{0, 1\}^d$ is the vector which is $1$ for coordinates in $S$. We say vector $v$ is $s$-sparse if at most $s$ of its entries are nonzero. For vectors $u$, $v$ of equal dimension, $u \circ v$ is their coordinatewise product. We let $\ball_2^d \defeq \{v \in \R^d \mid \norm{v}_2 \le 1\}$, and for sets $\set, \set' \subset \R^d$, $\set + \set'$ is their Minkowski sum. For convex sets $A, B \subseteq \R^d$, $N(A,B)$ is their covering number, the fewest number of translates of $B$ needed to cover $A$. We let $\Nor(\mu, \msig)$ be the Gaussian distribution of specified mean and covariance, and $\Nor_{\le \tau}(\mu, \msig)$ sets any draw from $\Nor(\mu, \msig)$ with Euclidean norm more than $\tau$ to the zero vector. We use ``with high probability in $n$'' to mean an event succeeds with probability $1 - n^{-O(1)}$, for an arbitrary constant.

\paragraph{Matrices.} Matrices are denoted in boldface. The $n \times n$ identity matrix and all-zeroes matrix are respectively $\id_n$ and $\mzero_n$. The number of nonzero entries of a matrix or a vector argument is $\nnz(\cdot)$. The set of $n \times n$ real symmetric matrices (and respectively, real positive semidefinite and positive definite matrices) is denoted $\Sym^n$ (and respectively, $\PSD^n$ and $\PD^n$). We equip $\Sym^n$ with the trace inner product $\inprod{\ma}{\mb} = \Tr(\ma\mb)$ and Loewner partial ordering $\preceq$. The Frobenius, operator, and trace norms are denoted $\normf{\cdot}$, $\normop{\cdot}$, and $\normtr{\cdot}$, and correspond to the $2$-norm, $\infty$-norm, and $1$-norm of the singular values of a matrix. For matrices $\{\ma_i\}_{i \in [m]}$ of equal dimension clear from context, we associate a pair of operators $\alla$ and $\alla^*$, taking a vector and matrix argument respectively. We let \begin{equation}\label{eq:alla_def}\alla(x) \defeq \sum_{i \in [m]} x_i \ma_i\text{ and } \alla^*(\my) \defeq \{\inprod{\ma_i}{\my}\}_{i \in [m]}.\end{equation}
We let $\Nor(\mu, \msig)$ denote the multivariate Gaussian with mean $\mu$ and covariance $\msig$, and $\gamma_d$ denotes the standard Gaussian density (of $\Nor(\0_d, \id_d)$) in dimension $d$. We let $\tmv(\ma)$ denote the time it takes to compute $\ma v$ for any vector $v$ of appropriate dimension. Note that when $\ma$ is explicit, $\tmv(\ma) = O(\nnz(\ma))$, where we let $\nnz(\ma)$ denote the number of nonzero entries in $\ma$. For any $\rho > 0$ and operator $\alla$ associated with a matrix set, we denote the associated ``discrepancy body'' by
\begin{equation}\label{eq:set_r_def}\set_{\rho, \alla} \defeq \{x \in \R^m \mid \normop{\alla(x)} \le \rho\}.\end{equation} 
We denote the condition number (ratio of largest and smallest eigenvalues) of $\mm \in \PD^n$ by $\kappa(\mm)$.

\paragraph{Graphs.} We use the notation $G = (V, E, w_G)$ to denote an undirected graph $G$ on vertex set $V$, where $w_G \in \R^E_{\ge 0}$ is the weight function on the edges $E$. When $G$ is clear from context we let $\mw \defeq \diag{w_G} \in \R^{E \times E}$ and we let $|\mb| \in \R^{E \times V}$ be the unsigned edge-vertex incidence matrix (i.e.\ with two ones per row corresponding to incident vertices on an edge). We also let $\mb \in \R^{E \times V}$ be the signed edge-vertex incidence matrix (where one of the ones per row is arbitrarily but consistently negated) and $\lap \defeq \mb^\top \mw \mb$ be the graph Laplacian. We refer to the rows of $\mb$ by $\{b_e\}_{e \in E} \in \{-1, 0, 1\}^V$. It is well-known that $x^\top \lap x = \sum_{e = (u, v) \in E} [w_G]_e (x_u - x_v)^2$. The pseudoinverse of $\lap$ (which has a kernel in the $\1_V$ direction, i.e.\ it preserves the kernel of $\lap$) is denoted $\lap^\dagger$. When multiple graphs are in discussion, we will subscript these matrices with $G$ to disambiguate, i.e.\ $\lap_G$ is the Laplacian of $G$. We denote the identity matrix on the subspace of $\R^V$ orthogonal to the $\1_V$ direction by $\idnoone$. %
\section{Approximate partial coloring framework}\label{sec:framework}

In this section, we recall the partial coloring method for discrepancy minimization from \cite{Rothvoss14}, and develop a binary search framework for approximating it efficiently. This framework will be our primary workhorse in all our applications. First, \cite{Rothvoss14} gave a generic result which transforms Gaussian measure lower bounds on a symmetric set $\set \subset \R^m$ to a high-probability lower bound on the number of ``tight constraints'' encountered by the closest point in $\set \cap [-1, 1]^m$ to a random Gaussian vector. We state a strengthening by \cite{RR22}, which works for any exponential measure lower bound, and also tolerates restriction to a linear subspace of sufficiently large dimension.

\begin{proposition}[Theorem 6, \cite{RR22}]\label{prop:many_tight}
Let $\set \subseteq \R^m$ be symmetric and convex, and suppose there is a constant $C$ such that $\gamma_m(\set) \ge \exp(-Cm)$. There are constants $\ctight \in (0, 1)$ and $\Cset > 0$ such that if $S \subseteq \R^m$ is a linear subspace of dimension at least $2\ctight m$, $g \sim \Nor(\0_m, \id_m)$, and $x_g$ is the closest point to $g$ in $\Cset \set \cap [-1, 1]^m \cap S$, then with probability $1 - \exp(-\Omega(m))$, at least $\ctight m$ coordinates of $x_g$ have magnitude $1$.
\end{proposition}

Proposition~\ref{prop:many_tight} is often referred to as a partial coloring result, as it fixes a number of coordinates of $x$ to $\pm 1$ (``colors''), and the remaining coordinates can be recursively handled. When the set $\set$ is a discrepancy body (as it will be throughout this paper), this recursive application then leads to discrepancy minimization algorithms. However, in this section we state our framework for a more general abstract setting, which may be of further utility. Our framework yields an \emph{approximate} partial coloring result, i.e.\ it will only guarantee that $x$ has many coordinates near $\pm 1$, but has the benefit that in settings of interest it can be implemented efficiently via customized solvers.

Specifically, we consider the setting where $\set$ is a sublevel set of a nonnegative, symmetric convex function $f: \R^m \to \R_{\ge 0}$ (i.e.\ $f(x) = f(-x)$ for all $x$), and we are given a ``discrepancy radius'' $\rho \ge 0$ fixed throughout. We then define our set of interest
\begin{equation}\label{eq:set_f_def}
\set \defeq \Brace{x \in \R^m \mid f(x) \le \rho}.
\end{equation}
We also assume that we are given a vector $g \in \R^m$ fixed throughout, which will ultimately be a random Gaussian vector. We use the following assumption which holds except with exponentially small probability $\exp(-\Omega(m))$ by standard Gaussian concentration.

\begin{assumption}\label{assume:rs_big}
We assume that $\norm{g}_2 \le 2\sqrt{m}$.
\end{assumption}

We are now ready to describe our binary search framework. In light of Proposition~\ref{prop:many_tight}, our goal is to approximately compute the value $\rs$ and point $\xs$ defined as
\begin{equation}\label{eq:rsdef}\rs \defeq \norm{\xs - g}_2, \text{ where } \xs \defeq \arg\min_{x \in [-1, 1]^m \cap\set } \norm{x - g}_2,\end{equation}
and $\set$ is defined as in \eqref{eq:set_f_def}. We remark that our framework extends to handle linear constraints as suggested by Proposition~\ref{prop:many_tight}, e.g.\ $\ma x = b$, by optimizing over $[-1, 1]^m \cap \{x \mid \ma x = b\}$. However, for simplicity of exposition we focus our attention to the case of optimizing over a hypercube in this section, and give a formal application of our framework which handles additional linear constraints in Section~\ref{sec:degree}. Our approximate partial coloring methods are developed as follows.

\begin{enumerate}
    \item We show $x$ approximating $\xs$ up to a distance $\approx \sqrt{m}\beta$ for an approximation tolerance $\beta > 0$, set differently in each specific instantiation, suffices for the partial coloring frameworks of e.g.\ \cite{Rothvoss14, RR20} up to a rounding error component (Lemma~\ref{lem:many_nearly_tight}). We typically handle this error recursively. We further exploit strong convexity of the squared distance function to show that approximately minimizing the objective in \eqref{eq:rsdef} suffices for this goal (Lemma~\ref{lem:sc_distsquare}). \label{item:step1_bs}
    \item We develop an efficient binary-search based reduction from approximating $\rs$ in \eqref{eq:rsdef} to approximating the values of a sequence of regularized convex optimization problems which balance $f(x)$ with a multiple of the squared distance to $g$ (Proposition~\ref{prop:binary_search_lam}).\label{item:step2_bs}
    \item We approximate the solutions of these regularized optimization problems in each of our applications using customized solvers. For sparsification-related settings, these solvers will usually be modifications of the recent box-spectraplex game solver of \cite{JambulapatiT23}.\label{item:step3_bs}
\end{enumerate}

Item~\ref{item:step3_bs} of our framework will be treated differently in each of our applications, so in this section we focus on providing results encapsulating Items~\ref{item:step1_bs} and~\ref{item:step2_bs} for ease of use.

\paragraph{Approximating $\xs$ suffices.}
We first observe that in applications of Proposition~\ref{prop:binary_search_lam}, any feasible point $x \in [-1, 1]^m \cap \set$ which is sufficiently close to $\xs$ must have many nearly-tight coordinates. In our applications we will apply Proposition~\ref{prop:binary_search_lam} to both $g$ and $-g$ for a random Gaussian $g$, and $f$ in \eqref{eq:set_f_def} will always be symmetric. So, without loss of generality we state our results when $\xs$ defined in \eqref{eq:rsdef} has at least a constant fraction of coordinates at $-1$, which will be convenient later.

\begin{lemma}\label{lem:many_nearly_tight}
Suppose $\xs \in [-1, 1]^m$ is $-1$ in at least $\frac {\ctight m} 2$ coordinates, and $x \in [-1, 1]^m$ satisfies
\[\norm{x - \xs}_2^2 \le \frac{\ctight m\beta^2}{4}.\]
Then, at least $\frac{\ctight m}{4}$ coordinates $x_i$ satisfy $x_i \le -1 + \beta$.
\end{lemma}
\begin{proof}
Assume otherwise, and let $I \subseteq [m]$ be the set of coordinates where $[\xs]_i = -1$. More than half of the coordinates $i \in I$ have $x_i \ge - 1 + \beta$ (since we assumed less than $\frac{\ctight m} 4$ coordinates of $x$ in total are below this threshold), which gives a contradiction since
\begin{align*}
\norm{x - \xs}_2^2 \ge \sum_{\substack{i \in I \\ x_i \ge -1 + \beta}} \Par{x_i - [\xs]_i}^2 > \frac{\ctight m \beta^2}{4}.
\end{align*}
\end{proof}
We next show that by strong convexity, we can certify a bound on $\norm{\tx - \xs}_2$ by finding any feasible point $\tx \in \set_{\rho} \cap [-1, 1]^m$ with near-minimal distance to $g$.

\begin{lemma}\label{lem:sc_distsquare}
Suppose $\tx \in \set \cap [-1, 1]^m$ and $\norm{\tx - g}_2^2 \le \rs^2 + \frac{\ctight m\beta^2}{4}$. 
Then $\norm{\tx - \xs}_2^2 \le \frac{\ctight m\beta^2}{4}$.
\end{lemma}
\begin{proof}
Recall that by definition, $\xs = \arg\min_{x \in \set \cap [-1,1]^m} \norm{x - g}_2^2$. 
Since $\norm{x - g}_2^2$ is $2$-strongly convex in the $\ell_2$ norm and $\tx \in \set \cap [-1, 1]^m$, optimality of $\xs$ implies the claim by rearranging
\begin{align*}
\norm{\tx - \xs}_2^2 + \rs^2 = \norm{\tx - \xs}_2^2 + \norm{\xs - g}_2^2 \le \norm{\tx - g}_2^2 \le \rs^2 + \frac{\ctight m\beta^2}{4}.
\end{align*}
\end{proof}

In light of Lemma~\ref{lem:sc_distsquare}, we find a point in $[-1, 1]^m \cap \set$ with near-minimal distance to $g$ by repeatedly solving optimization problems which balance the objectives $f(x)$ and $\norm{x - g}_2^2$, described next.

\paragraph{Testing via regularized optimization.} Our algorithm is based on a binary search subroutine, each iteration of which is based on approximately solving an optimization problem $f_\lam$ parameterized by a regularization parameter $\lam$, belonging to the following family:
\begin{equation}\label{eq:test_opt}
\min_{x \in [-1, 1]^m}f_\lam(x) \defeq f(x) + \lam\norm{x - g}_2^2.
\end{equation}
We will discuss how to efficiently approximate the optimizer of $f_\lam$ in \eqref{eq:test_opt} using tools from \cite{JambulapatiT23} and other optimization subroutines in the other sections of this paper. Here, we show how approximate solutions to \eqref{eq:test_opt} can be used to perform a binary search on $\lam$ and approximate $\xs$. Throughout, we will denote our desired additive error threshold (as specified in Lemma~\ref{lem:sc_distsquare}) by
\[\tau \defeq \frac{\ctight m\beta^2}{4}.\]
We will also denote for convenience,
\begin{equation}\label{eq:xlamdef}x_g \defeq \arg\min_{x \in [-1, 1]^m} \norm{x - g}_2,\; x_\lam \defeq \arg\min_{x \in [-1, 1]^m} f_\lam(x)\text{ for all } \lam \ge 0. \end{equation}
We begin by establishing crude bounds on our $\lam$ binary search, assuming the bounds
\begin{equation}\label{eq:rangedef}
\Theta \ge f(x) \ge 0 \text{ for all } x \in [-1, 1]^m,\text{ and } f(\0_m) = 0.
\end{equation}
We remark that because $f$ is symmetric and convex, it is minimized at $\0_m$, and in all our applications we have $f(\0_m) = 0$; otherwise, one may first perform an additive shift to $f$ to use our framework.

\begin{lemma}\label{lem:lam_upper}
Suppose $f$ satisfies \eqref{eq:rangedef}.
For $\lam \ge \frac {4\Theta} \tau$, $x_g$ minimizes $f_\lam$ to additive error $\frac {\lam\tau}{4}$ over $[-1, 1]^m$.
\end{lemma}
\begin{proof}
Defining $x_\lam$, $x_g$ as in \eqref{eq:xlamdef}, and using optimality of $x_g$ and \eqref{eq:rangedef},
\[
f_\lam(x_g) \le \Theta + \lam\norm{x_g - g}_2^2 \le \frac{\lam\tau}{4} + f(x_\lam) + \lam\norm{x_\lam - g}_2^2.
\]
\end{proof}

\begin{lemma}\label{lem:lam_lower}
Suppose $f$ satisfies \eqref{eq:rangedef}. For $\lam \le \frac \rho {8m}$, and under Assumption~\ref{assume:rs_big}, any $x$ minimizing $f_\lam$ to additive error $\frac \rho 4$ over $[-1, 1]^m$ has $f(x) \le \frac{3\rho}{4}$.
\end{lemma}
\begin{proof}
Since $\0_m \in [-1, 1]^m$ and $f(\0_m) = 0$, we have by the assumption on $\norm{g}_2 \le 2\sqrt{m}$ that
\begin{align*}
f(x) \le \lam\Par{\norm{g}_2^2 - \norm{x - g}_2^2} + \frac \rho 4 \le 4\lam m + \frac \rho 4 \le \frac{3\rho} 4.
\end{align*}
\end{proof}

We also require a helper lemma that rounds a near-feasible point with bounded distance increase.

\begin{lemma}\label{lem:drag_down_opnorm}
Suppose $f$ satisfies \eqref{eq:rangedef}. If $x \in [-1, 1]^m$ satisfies $f(x) \le \rho(1 + c)$, then defining $\tx \gets \frac{1}{1+c}x$, under Assumption~\ref{assume:rs_big} we have $\tx \in \set \cap [-1, 1]^m$ and $\norm{\tx - g}_2^2 \le \norm{x - g}_2^2 + 4cm$.
\end{lemma}
\begin{proof}
The first conclusion follows by definition of $\tx$ and convexity of $f$, where we use $f(\0_m) = 0$. The second uses $\norm{g}_2\norm{x}_2 \le 2m$, and so we may bound
\[\norm{\tx - g}_2^2 - \norm{x - g}_2^2 = \frac{2c}{1+c}\inprod{g}{x} - \Par{1 - \Par{\frac 1 {1 + \alpha}}^2}\norm{x}_2^2 \le 4cm.\]
\end{proof}

Finally, we show that for nearby values of $\lam$, we can aggregate approximate minimizers of $f_\lam$ to obtain a point in $\set \cap [-1, 1]^m$ with nearly-optimal distance to $g$. In the statement of the following result, we assume only multiplicative error guarantees on querying the value of $f$. This is because in our applications where $f(x) = \normop{\alla(x)}$, a full eigendecomposition is prohibitively expensive so we will approximate operator norms via the power method (see Proposition~\ref{prop:krylov}).

\begin{lemma}\label{lem:aggregate_approx}
Let $\lam' \in (\lam, (1 + \frac \tau {10m})\lam)$, and suppose $x$, $x'$ respectively minimize $f_{\lam}$ and $f_{\lam'}$ to additive error $\frac{\lam'\tau}{4}$ over $[-1, 1]^m$. Further assume for $c \defeq \min(\frac \tau {12m}, \frac{\tau\lam}{3\rho})$ that we know scalar values $A \le \rho$ and $A' \ge (1+c)\rho$ such that $f(x) \in [(1-c)A, A]$ and $f(x') \in [(1-c)A', A']$. Define $\alpha \defeq 1 - \frac{\rho - A}{A' - A}$ and $\tx \defeq \alpha x + (1-\alpha)x'$. Then under Assumption~\ref{assume:rs_big},
\[\norm{\tx - g}_2^2 \le \rs^2 + \tau\text{ and } \tx \in [-1, 1]^m \cap \set.\]
\end{lemma}
\begin{proof}
By convexity of the set $[-1, 1]^m$ and the function $f$, and because $\alpha A + (1 - \alpha)A' = \rho$ by construction, we clearly have $y \in [-1, 1]^m \cap \set$, so it suffices to show the distance bound. Recall from \eqref{eq:rsdef} that there exists $\xs \in [-1, 1]^m \cap \set$ with $\norm{\xs - g}_2^2 = \rs^2$ by definition. Since
\begin{align*}
(1-c)A + \lam\norm{x - g}_2^2 \le f(x) + \lam\norm{x - g}_2^2 &\le f(\xs) + \lam\rs^2 + \frac{\lam'\tau}{4} \le \rho + \lam\rs^2 + \frac{\lam'\tau}{4}, \\
(1-c)A' + \lam\norm{x' - g}_2^2 \le f(x') + \lam'\norm{x' - g}_2^2 &\le f(\xs) + \lam'\rs^2 + \frac{\lam'\tau}{4} \le \rho + \lam'\rs^2 + \frac{\lam'\tau}{4},
\end{align*}
we have by a convex combination of the above display that, for $\blam \defeq \alpha\lam + (1 - \alpha)\lam'$,
\begin{align*}
(1-c)\rho + \alpha\lam\norm{x - g}_2^2 + (1 - \alpha)\lam'\norm{x - g}_2^2 \le \rho + \blam \rs^2 + \frac{\lam'\tau}{4}.
\end{align*}
Rearranging and dividing by $\blam \in (\lam, \lam')$,  we have
\begin{align*}
\norm{\tx - g}_2^2 &\le \rs^2 + \frac{c\rho}{\blam} + \frac{\lam'\tau}{4\blam} + \frac{1}{\blam}\Par{\blam\norm{\tx - g}_2^2 - \alpha\lam\norm{x - g}_2^2 - (1-\alpha)\lam'\norm{x' - g}_2^2} \\
&\le \rs^2 + \frac{\tau}{3} + \frac \tau 3 + \frac{1}{\blam}\Par{(\blam - \lam)\alpha\norm{x - g}_2^2} \le \rs^2 + \frac {2\tau} 3   + \frac {10m(\lam' - \lam)} {\lam} \le \rs^2 + \tau.
\end{align*}
The second inequality used $\norm{\tx - g}_2^2 \le \alpha\norm{x - g}_2^2 + (1 - \alpha)\norm{x' - g}_2^2$ by convexity as well as our bounds on $c$ and $\frac{\lam'}{\lam}$, and the third used $\norm{x - g}_2^2 \le 2\norm{x}_2^2 + 2\norm{g}_2^2 \le 10m$ by Assumption~\ref{assume:rs_big}.
\end{proof}

We combine these developments to give our final procedure for approximate partial coloring.

\begin{proposition}\label{prop:binary_search_lam}
Let $\beta \in (0, 1)$, $\rho > 0$, and let $\tau = \Theta(m\beta^2)$, $T = O(\log \frac 1 \beta + \log\log \frac \Theta {\rho})$ for appropriate constants. Let symmetric and convex $f: \R^m \to \R_{\ge 0}$ satisfy \eqref{eq:rangedef}, and define $\set$ as in \eqref{eq:set_f_def}. Under Assumption~\ref{assume:rs_big} and following notation \eqref{eq:rsdef}, there is an algorithm returning $\tx \in \set \cap [-1, 1]^m$ with
\[\norm{\tx - g}_2^2 \le \rs^2 + \tau.\]
The runtime is the cost of solving $f_\lam$ to additive error $\frac{\lam\tau}{4}$ for $T$ values of $\lam \ge \frac{\rho}{8m}$, and computing $f(x)$ to within $\frac{\tau}{64m}$ multiplicative error for $T$ values of $x \in [-1, 1]^m$.
\end{proposition}
\begin{proof}
Let $c \defeq \frac{\tau}{64m} = \Theta(\beta^2)$, and divide the range $[\frac{\rho}{8m},\frac {4\Theta} \tau]$ into multiplicative intervals of the form $[\lam, (1 + c)\lam]$; there are clearly $O(\frac 1 c \log \frac{m\Theta}{\tau\rho}) = O(\frac 1 {\beta^2} \log\frac{\Theta}{\beta\rho})$ such intervals. We will initialize $\lam \gets \frac \rho {8m}$ and $\lam' \gets \frac {4\Theta} \tau$, and throughout our binary search we maintain the following invariants.
\begin{enumerate}
    \item $\lam, \lam'$ are lower and upper endpoints of different multiplicative intervals.
    \item $x$ and $x'$ respectively minimize $f_\lam$ and $f_{\lam'}$ to additive error $\frac{\lam\tau}{4}$ and $\frac{\lam'\tau}{4}$ over $[-1, 1]^m$.
    \item $A$ and $A'$ respectively satisfy $f(x) \in [(1 - c)A, A]$ and $f(x') \in [(1 - c)A', A']$.
    \item $A \le \rho$ and $A' \ge (1 + c)\rho$.
\end{enumerate}
At initialization, since $\frac{\tau\rho}{32m} \le \frac \rho 4$, we can set $x$ to any $\frac{\tau\rho}{32m}$-approximate minimizer of $f_{\lam}$ and $A$ to be any value satisfying $f(x) \in [(1 - c)A, A]$, since $A$ will necessarily be less than $\rho$ by Lemma~\ref{lem:lam_lower}. Further, Lemma~\ref{lem:lam_upper} shows that it suffices to initialize $x' \gets x_g$, and let $A'$ be any value satisfying $\normop{\alla(x')} \in [(1 - c)A', A']$. If $A' \le (1 + c)\rho$, we can apply Lemma~\ref{lem:drag_down_opnorm} to $x_g$ yielding $\tx$ with 
\begin{equation}\label{eq:txbound}\norm{\tx - g}_2^2 \le \norm{x_g - g}_2^2 + 4cm \le \rs^2 + \frac \tau {16},\; \tx \in \set_\rho \cap [-1, 1]^m.\end{equation}
Hence it suffices to return $\tx$. Otherwise, $A' \ge (1 + c)\rho$ so all invariants are met at initialization. Next, let $\lamt$ be the lower endpoint of the middle interval between $\lam$ and $\lam'$, let $\xt$ be a $\frac{\lamt \tau}{4}$-approximate minimizer to $f_{\lamt}$, and let $f(\xt) \in [(1 - c)\At, \At]$. If $\At \in [\rho, (1 + c)\rho]$,
\begin{align*}
f(\xt) + \lamt\norm{\xt - g}_2^2 &\le f(\xs) + \lamt\rs^2 + \frac{\lamt \tau}{4} \\
&\le \At + \lamt \rs^2 + \frac{\lamt \tau}{4} \\
&\le \frac{1}{1-c} f(\xt) + \lamt \rs^2 + \frac{\lamt \tau}{4},
\end{align*}
so rearranging and dividing by $\lamt \ge \frac{\rho}{8m}$ proves that 
\[\norm{\xt - g}_2^2 \le \frac{c}{1-c} \cdot \frac{(1 + c)\rho}{\lamt} + \rs^2 + \frac \tau 4 \le 2c\rho \cdot \frac{8m}{\rho} + \rs^2 + \frac \tau 4 \le \rs^2 + \frac \tau 2, \]
where we used our choice of $c$. Therefore, letting $\tx$ be the result of applying Lemma~\ref{lem:drag_down_opnorm} to $\xt$ with parameter $c$, the same calculation as \eqref{eq:txbound} shows 
\[\norm{\tx - g}_2^2 \le \rs^2 + \tau,\; \tx \in \set \cap [-1, 1]^m.\]
Otherwise, we either have $\At \le \rho$ or $\At \ge (1 + c)\rho$, so we may update $\lam$ or $\lam'$ to $\lamt$ and continue. At termination we must have $\lam$ and $\lam'$ within a multiplicative factor of $1 + c$, and therefore applying Lemma~\ref{lem:aggregate_approx} yields the desired point $\tx$, where we check that indeed $c \le \min(\frac{\tau}{12m}, \frac{\tau}{3\rho} \cdot \frac{\rho}{8m})$.

Finally, we test $T = O(\log(\frac{1}{\beta^2}\log \frac \Theta {\beta\rho}))$ values of $\lam$, and the runtime is from running approximate minimization routines to maintain $x$, $x'$ and approximate function calls to maintain $A$, $A'$.
\end{proof}

In a typical application of the framework of this section, we will first specify a discrepancy radius $\rho$ such that the body $\set$ from \eqref{eq:set_f_def} has an exponential Gaussian measure lower bound. We then can combine Proposition~\ref{prop:many_tight}, Lemma~\ref{lem:many_nearly_tight}, Lemma~\ref{lem:sc_distsquare}, and Proposition~\ref{prop:binary_search_lam} to efficiently obtain an approximate partial coloring, by approximating $\xs$ for both $g$ and $-g$ (in the notation of \eqref{eq:rsdef}). Both of our approximations via Proposition~\ref{prop:binary_search_lam} will lie in $[-1, 1]^m \cap \set$, and then the other three results imply that one of the approximations has a constant fraction of coordinates in $[-1, -1+\beta]$. %
\section{Gaussian measure lower bounds}\label{sec:measure}

In this section, we give two new tools for obtaining Gaussian measure lower bounds, a key ingredient of the \cite{Rothvoss14} framework. In Section~\ref{ssec:reduction}, we first state a general reduction from bounding Gaussian measure for a set $\set$ to proving a closely related notion introduced by \cite{RR20}, which roughly says that a random Gaussian vector is not too far from a multiple of $\set$ with constant probability. This can be seen as a generalization of Conjecture 1 in~\cite{RR20}. We complement this reduction with a strengthening of the main result of \cite{RR20}, which is parameterized by a trace bound, and is needed to prove our ultrasparsifier results in Section~\ref{sec:ultra}. Finally, in Section~\ref{ssec:gaussian_lbs} we prove as corollaries of these tools the main Gaussian lower bounds we will use in the paper.

\subsection{Reduction to Gaussian distance bound}\label{ssec:reduction}

In this section, we prove the following reduction holds.

\begin{theorem}\label{thm:expansionimplieslowerbound}
Suppose $\set \subset \R^m$ is symmetric and convex, and that for a constant $C_0$,
\[\gamma_m\Par{\frac{C_0}{\alpha}\set + \alpha\sqrt{m}\ball_2^m} \ge \half \text{ for all } \alpha \in (0, 1).\]
Then there is a constant $C$ such that $\gamma_m(\set) \ge \exp(-Cm)$.
\end{theorem}

The main technical tool we use is the existence of \textit{$1$-regular $M$-ellipsoids}. We state a variant of  \cite{pisier_1989}, which gives a result for $p$-regular $M$-ellipsoids for all $p \in (0, 2)$; we only use $p = 1$.

\begin{proposition}[Corollary 7.16,~\cite{pisier_1989}] \label{prop:regmellipsoid}
    There is a constant $C_1$ such that for any symmetric convex $\set \subseteq \R^m$, there is an ellipsoid $\ME$ with 
    \[ \max\{N(\set,t\ME), N(\ME, t\set)\} \le \exp(C_1 m t^{-1}).\]
\end{proposition}

Before giving our main technical result Lemma~\ref{lem:shortaxes}, we need a few facts about covering numbers. For the following standard facts, see Theorem 4.1.13 and Facts 4.1.7, 4.1.8 and 4.1.9 in~\cite{AsymptoticGeometricAnalysis-Book2015}.

\begin{fact} \label{fact:CoverMyCover}
For any convex sets $A,B,C \subseteq \R^m$, we have $N(A,B) \le N(A,C) N(C, B)$.
\end{fact}

\begin{fact} \label{fact:CoverMySum}
For any convex sets $A,B,C \subseteq \R^m$, we have $N(A+C,B+C) \le N(A,B)$.
\end{fact}

\begin{fact}\label{fact:coveringVolLB}
 For any convex sets $A,B \subseteq \R^m$, we have $N(A,B) \ge \frac{\Vol(A)}{\Vol(B)}$. Similarly, $N(A,B) \ge \frac{\gamma_m(A)}{\gamma_m(B)}$.
\end{fact}

\begin{fact}\label{fact:coveringVolUB}
 For any convex sets $A,B \subseteq \R^m$ with $B$ symmetric, $N(A,B) \le \frac{\Vol(2A + B)}{\Vol(B)}$.   
\end{fact}

\begin{fact}\label{fact:coveringIntersection}
   For any convex sets $A,B \subseteq \R^m$ with $A$ symmetric, $N(A,2(A \cap B)) \le N(A,B)$.     
\end{fact}

For the following fact, see Lemma 3.3 in~\cite{DJR22}.

\begin{fact}\label{fact:CoverBallLargeSet}
For any symmetric convex set $\set \subseteq \R^m$ with $\gamma_m(\set) \ge \frac{1}{2}$, we have $N(\sqrt{m} \ball^m_2, \set) \le \exp(C_2 m)$ for a universal constant $C_2 > 0$.    
\end{fact}

As a corollary of Fact~\ref{fact:coveringVolUB} we also have the following.

\begin{lemma}\label{lem:coveringIntersectVolUB}
    For symmetric convex sets $A, B \subseteq \R^m$, $N(A,B) \le 3^m \cdot \frac{\Vol(A)}{\Vol(A \cap B)}$.
\end{lemma}

\begin{proof}
    Since $A \cap B$ is also symmetric, by Fact~\ref{fact:coveringVolUB} we have
    \[ N(A,B) \le N(A, A \cap B) \le \frac{\Vol(2A + (A \cap B)))}{\Vol(A \cap B)} \le \frac{\Vol(3A)}{\Vol(A \cap B)} = 3^m \cdot \frac{\Vol(A)}{\Vol(A \cap B)}. \qedhere \]
\end{proof}

We will also need a more specific result about covering numbers of slices by subspaces.

\begin{lemma}\label{lem:subspaceCover}
    For any symmetric convex $A, B \subseteq \R^m$ and any $d$-dimensional subspace $U \subseteq \R^m$, $N(A \cap U, B \cap U) \le 6^d \cdot N(A,B)$.
\end{lemma}

\begin{proof}
By Lemma~\ref{lem:coveringIntersectVolUB}, we have
\[ N(A \cap U, B \cap U) \le N(A \cap U, A \cap B \cap U) \le 3^d \cdot \frac{\Vol(A\cap U)}{\Vol(A \cap B \cap U)}.\]
It remains to note that by Fact~\ref{fact:coveringIntersection}, we can cover $A$ with $N(A,B)$ copies of $2(A \cap B)$, so that Fact~\ref{fact:coveringVolLB} yields $\Vol(A \cap U) \le N(A, B) \cdot \Vol(2(A\cap B) \cap U) = 2^d \cdot N(A,B) \cdot \Vol(A \cap B \cap U)$.
\end{proof}

We are now ready to state and prove our main technical lemma.

\begin{lemma} \label{lem:shortaxes}
Let $\set \subseteq \R^m$ be a symmetric convex set such that $\gamma_m (\frac{C_0}{\alpha} \set + \alpha \sqrt{m} \ball^m_2 )\ge \frac{1}{2}$ for all $\alpha > 0$ and let $\ME$ be a 1-regular $M$-ellipsoid for $\set$. Then there is a universal constant $C'$ such that for every $r > 0$, the number of axes of $\ME$ of length at most $r$ is at most $C' \sqrt r \cdot m^{\frac 3 4}$.
\end{lemma}
\begin{proof}
Let $U$ denote the span of directions corresponding to the axes of $\ME$ of length at most $r$ with $d := \dim(U)$. To simplify notation, let $\set(\alpha) := \frac{C_0}{\alpha} \set + \alpha \sqrt{m} \ball^m_2$, so that we still have $\gamma_U (\set(\alpha) \cap U) \ge \gamma_m (\set(\alpha)) \ge \frac{1}{2}$, where $\gamma_U$ is the $d$-dimensional Gaussian measure in $U$. In particular, letting $\ball^U_2$ denote the unit Euclidean ball restricted to $U$, $N(\sqrt{d} \ball^U_2, \set(\alpha) \cap U) \le \exp(C_2 d)$ by Fact~\ref{fact:CoverBallLargeSet}.

Further, note that for any $\alpha > 0$ we have
\begin{align*}
N\Big(\set(\alpha) \cap U, \Big(\frac{C_0}{\alpha} \cdot \frac{m r}{d} + \alpha \sqrt{m}\Big) \ball^U_2\Big) &\le N\Big(\set(\alpha) \cap U, \Big(\frac{C_0}{\alpha} \cdot \frac{m}{d} \ME + \alpha \sqrt{m} \ball^m_2\Big) \cap U\Big) \\ &\le 6^d \cdot N\Big(\set(\alpha), \frac{C_0}{\alpha} \cdot \frac{m}{d} \ME + \alpha \sqrt{m} \ball^m_2\Big) \\ & \le 6^d \cdot N(\set, \frac{m}{d} \ME)  \le 6^d \cdot \exp(C_1 d),
\end{align*}
where in the second inequality we use Lemma~\ref{lem:subspaceCover}, in the third we use Fact~\ref{fact:CoverMySum} and in the fourth we use Proposition~\ref{prop:regmellipsoid}. Setting $\alpha := \sqrt{C_0 \frac{mr}{d} } \cdot m^{-\frac 1 4}$, we have $\frac{C_0}{\alpha} \cdot \frac{m r}{d} + \alpha \sqrt{m} =  \sqrt{C_0 \frac{r}{d} } \cdot m^{\frac 3 4}$, so by Fact~\ref{fact:CoverMyCover},
\begin{align*} N\Par{\sqrt{d} \ball^U_2, 2\sqrt{ \frac{C_0 r}{d}} m^{\frac 3 4} \ball^U_2} & \le N\Par{\sqrt{d} \ball^U_2, \set(\alpha) \cap U} \cdot N\Par{\set(\alpha) \cap U,  2 \sqrt{\frac{C_0 r}{d}} m^{\frac 3 4} \ball^U_2} \\ & \le \exp(C_2 d) \cdot 6^d \cdot \exp(C_1 d). \end{align*}

On the other hand, by Fact~\ref{fact:coveringVolLB},
\[N\Par{\sqrt{d} \ball^U_2, 2\sqrt{\frac{C_0 r}{d}} m^{\frac 3 4} \ball^U_2}\ge \frac{\Vol_d (\sqrt{d} \ball^U_2)}{\Vol_d (2\sqrt{ \frac{C_0 r}{d}} m^{\frac 3 4} \ball^U_2)} = \Par{\frac{d}{2 \sqrt{C_0 r} m^{\frac 3 4} }}^d. \]
Combining the above two displays yields the claim.
\end{proof}

\begin{proof}[Proof of Theorem~\ref{thm:expansionimplieslowerbound}]
First we show a Gaussian measure lower bound for a 1-regular $M$-ellipsoid $\ME$ of $\set$ with axes of lengths $\{\lam_i\}_{i \in [m]}$ sorted in increasing order. Let $k$ denote the maximum index with $\lambda_k \le \sqrt{m}$ and let $\ME'$ denote the ellipsoid with the same eigenvectors as $\ME$ and compressed axes of length $\min\{\lambda_i, \sqrt{m}\}$, so that in particular $\ME' \subseteq \ME$. Note that
\begin{align*} \gamma_m (\ME) \ge \gamma_m(\ME') \ge \int_{\ME'} \frac{1}{(2\pi)^{m/2}} \exp(-\tfrac{1}{2} \underbrace{\|x\|_2^2}_{\le m}) \dd x \ge \exp(-C'm) \cdot \Vol (\ME')  & \ge \exp(-C'' m) \prod_{i\in[k]}\frac{\lambda_i}{\sqrt{m}}, \end{align*}
for some constants $C',C'' > 0$. Denote $I_r := \{i \in [k] : \lambda_i \in [\frac{r}{2}, r]\}$. We apply Lemma~\ref{lem:shortaxes} as follows:
\[  \prod_{i \in [k]} \frac{\lambda_i}{\sqrt{m}} \ge \prod_{r \in \sqrt{m} \cdot 2^{\mathbb{Z}_{\le 0} }  } \Big(\frac{r}{2\sqrt{m}}\Big)^{|I_r|} \ge  \prod_{r \in \sqrt{m} \cdot 2^{\mathbb{Z}_{\le 0} }  } \Big(\frac{r}{2\sqrt{m}}\Big)^{C'\sqrt r m^{\frac 3 4}} .\]
Denoting $\frac r {2\sqrt m} = 2^{-a}$, the product becomes, for a constant $C_3$,
\[ \prod_{i \in [k]} \frac{\lambda_i}{\sqrt{m}} \ge \Par{\prod_{a = 0}^\infty (2^{-a})^{2^{-\frac a 2}} }^{C' m} = \exp(-C_3 m).\]
Finally, we may use Fact~\ref{fact:coveringVolLB} and Proposition~\ref{prop:regmellipsoid} to lower bound the Gaussian measure of $\set$:
\[ \gamma_m (\set) \ge \frac{\gamma_m (\ME)}{N(\ME, \set)} \ge \exp(-(C'' + C_1 + C_3) m) . \qedhere\]
\end{proof}

\subsection{Operator norm discrepancy bodies}\label{ssec:gaussian_lbs}

In this section, we combine Theorem~\ref{thm:expansionimplieslowerbound} with the following generalization of the main result in~\cite{RR20} to obtain new Gaussian measure lower bounds. The argument we use to prove Theorem~\ref{thm:gen_rr20} is almost identical to \cite{RR20}, where we simulate a random walk while blocking ``dangerous'' directions to control a potential function, which results in a measure lower bound. The main difference is in the choice of the potential function we use to capture our trace condition. Due to the significant overlap with \cite{RR20}, we defer a proof of the following result to Appendix~\ref{app:modified_measure} for completeness.

\begin{restatable}{theorem}{restategenrr}\label{thm:gen_rr20}
Suppose $\{\ma_i\}_{i \in [m]} \subset \PSD^n$ satisfy $\alla(\1_m) \preceq \id_n$ and $\Tr(\alla(\1_m)) \le \tau$ for sufficiently large $m$ and $n \le 2^{\frac m 5}$. Following notation \eqref{eq:set_r_def}, there is a constant $C_0$ such that
\[\gamma_m\Par{\frac{C_0}{\alpha} \set_{\sqrt{\frac \tau m}, \alla} + \alpha \sqrt m \ball_2^m} \ge \half, \text{ for all } \alpha \in (0, 1).\]
\end{restatable}

As a first application of our techniques, we prove the following new Gaussian measure lower bound on the operator norm ball, which shows Conjecture 1 of \cite{RR20} is true.

\begin{restatable}{corollary}{restategaussianmeasure}\label{cor:gaussian_measure}
Suppose $\{\ma_i\}_{i \in [m]} \subset \PSD^n$ satisfy $\alla(\1_m) \preceq \id_n$, and $\frac m n$ is at least a sufficiently large constant. Following notation \eqref{eq:set_r_def}, there is a constant $C > 0$ such that
\[\gamma_m\Par{\set_{\sqrt{\frac n m}, \alla}} \ge \exp\Par{-Cm}.\]
\end{restatable}
\begin{proof}
It suffices to directly combine Theorem~\ref{thm:expansionimplieslowerbound} and Theorem 1 from \cite{RR20}.
\end{proof}

We also prove the following variant of Corollary~\ref{cor:gaussian_measure}, which will be used to design ultrasparsifiers.

\begin{restatable}{corollary}{restateultrameasure}\label{cor:ultra_measure}
Suppose $\{\ma_i\}_{i \in [m]} \subset \PSD^n$ satisfy $\alla(\1_m) \preceq \id_n$ and $\Tr(\alla(\1_m)) \le \tau$ for sufficiently large $m$ and $n \le 2^{\frac m 5}$. Following notation \eqref{eq:set_r_def}, there is a constant $C > 0$ such that
\[\gamma_m\Par{\set_{\sqrt{\frac \tau m}, \alla}} \ge \exp\Par{-Cm}.\]
\end{restatable}
\begin{proof}
It suffices to directly combine Theorem~\ref{thm:expansionimplieslowerbound} and Theorem~\ref{thm:gen_rr20}.
\end{proof} %
\section{Linear-sized sparsifiers}\label{sec:bss}

In this section we will be concerned with the following problem: we are given a set of matrices $\{\mm_i\}_{i \in [m]} \subset \PSD^n$ satisfying $\allm(\1_m) = \id_n$ where $\allm$ is the operator associated with the set in the sense of \eqref{eq:alla_def}, and a parameter $\eps \in (0, 1)$ (both fixed throughout). Our goal is to construct a vector $w \in \R^m_{\ge 0}$ such that $\nnz(w) = O(\frac n {\eps^2})$, and $\normop{\allm(w) - \id_n} \le \eps$. We will achieve this goal by an instantiation of the framework in Section~\ref{sec:framework}, combined with recursive rounding. We state the main regularized optimization and approximate function query subroutines we use (as required by Proposition~\ref{prop:binary_search_lam}) in Section~\ref{ssec:fast_gaussian_rounding}, and put the pieces together to prove our main result in Section~\ref{ssec:main-bss}. Finally, we discuss implications and runtime considerations for sparsifying graphs in Section~\ref{ssec:graph_bss}.

\subsection{Approximate Gaussian rounding in nearly-linear time}\label{ssec:fast_gaussian_rounding}

In the applications of the framework in Section~\ref{sec:framework} to the sparsification setting, we will define a discrepancy body with respect to $f(x) \defeq \normop{\alla(x)}$, where $\alla$ is the operator associated to a family of matrices $\{\ma_i\}_{i \in [m]}$. The matrix family will always be a reweighting of a subset of our original matrices $\{\mm_i\}_{i \in [m]}$. To leverage Proposition~\ref{prop:binary_search_lam} to approximate Gaussian roundings, we require two helper tools.
We first recall a guarantee on approximate top singular value computation.

\begin{proposition}[Theorem 1, \cite{MuscoM15}]\label{prop:krylov}
For any $c \in (0, 1)$ and $x \in \R^m$, we can return a scalar value $A$ such that $\normop{\alla(x)} \in [(1-c)A, A]$ with probability $\ge 1 - \delta$ in time
\[O\Par{\Par{\sum_{i \in [m]} \tmv(\ma_i)} \cdot \frac{\log \frac m \delta}{\sqrt c}}.\]
\end{proposition}

Proposition~\ref{prop:binary_search_lam} also requires approximate solvers for regularized optimization problems of the form $f_\lam$ defined in \eqref{eq:test_opt}. When $f(x) = \normop{\alla(x)}$, we have the characterization
\begin{equation}\label{eq:primaldual_opnorm}
\begin{gathered}
f_\lam(x) = \normop{\alla(x)} + \lam\norm{x - g}_2^2 = \max_{\my \in \Delta^{2n \times 2n}} \inprod{\my}{\sum_{i \in [m]} x_i \tma_i} + \lam \norm{x - g}_2^2, \\
\text{where } \tma_i \defeq \begin{pmatrix} \ma_i & \mzero_n \\ \mzero_n & -\ma_i \end{pmatrix},\text{ and } \Delta^{2n \times 2n} \defeq \{\my \in \PSD^{2n} \mid \Tr \my = 1\}.
\end{gathered}
\end{equation}

The recent work \cite{JambulapatiT23} developed a generic framework for solving box-spectraplex games, which \eqref{eq:primaldual_opnorm} is almost a case of; the main difference is that \eqref{eq:primaldual_opnorm} is regularized. We prove the following modification of Theorem 3 of \cite{JambulapatiT23} in Appendix~\ref{app:modified_boxspec}, which handles the extra part of the objective.

\begin{proposition}\label{prop:modified_boxspec}
Let $\delta, \Delta \in (0, 1)$, $\lam \ge 0$, and $g \in \R^m$. There is an algorithm which returns $\tx \in [-1, 1]^m$ minimizing $f_\lam$ (defined in \eqref{eq:primaldual_opnorm}) to additive error $\Delta$, with probability $\ge 1 - \delta$, in time
\[O\Par{\Par{\sum_{i \in [m]}\tmv(\ma_i)} \cdot \frac{\normop{\alla(\1_m)}^{3.5}\log^2\Par{\frac{mn}{\delta\Delta}}\log(n)}{\Delta^{3.5}}}.\]
\end{proposition}

\subsection{Linear-sized sparsifiers in nearly-linear time}\label{ssec:main-bss}

Finally, we put the pieces together to prove our main result. To warm start our algorithm, we use the following standard random sampler which loses a logarithmic factor from the desired sparsity.

\begin{lemma}\label{lem:lose_log}
Let $\{\mm_i\}_{i \in [m]} \subset \PSD^n$ satisfy $\allm(\1_m) \preceq \id_n$ and $\Tr(\allm(\1_m)) \le \tau$, and let $K = \Theta( \frac{\tau}{\eps^2} \log \frac n \delta)$ for an appropriate constant. Define independently distributed random matrices $\{\mx_k\}_{k \in [K]}$ such that
\[\mx_k = \frac 1 {Kp_i} \mm_i \text{ with probability } p_i \defeq \frac{\Tr(\mm_i)}{\sum_{i \in [m]} \Tr(\mm_i)} \text{ for all }k \in [K].\]
Then with probability $\ge 1 - \delta$, $\normsop{\sum_{k \in [K]} \mx_k - \allm(\1_m)} \le \eps$.
\end{lemma}
\begin{proof}
First, it is clear that $\{p_i\}_{i \in [m]}$ is a valid sampling distribution. Next we define random matrices $\my_k \defeq \frac 1 {K p_i} \mm_i - \frac 1 K \allm(\1_m)$, which are mean-zero and satisfy $\normop{\my_k} \le \frac \tau K$. Further, recall
\begin{align*}
(\ma - \mb)^2 \preceq 2\ma^2 + 2\mb^2,\; \frac 1 {\Tr \ma} \ma^2 \preceq \ma,
\end{align*}
for all $\ma, \mb \in \PSD^n$. The first inequality follows from $(\ma + \mb)^2 \in \PSD^n$ and the second follows because the two sides commute and we compare eigenvalues. Using these claims and $\allm(\1_m) \preceq \id_n$, we have
\begin{align*}
\sum_{k \in [K]} \E \my_k^2 &= \frac 1 K \sum_{i \in [m]} p_i \Par{\frac 1 {p_i} \mm_i - \allm(\1_m)}^2 \\
&\preceq \frac 2 K \sum_{i \in [m]} \frac 1 {p_i} \mm_i^2 + \frac 2 K \allm(\1_m)^2 \\
&\preceq \frac {2\tau} K \sum_{i \in [m]} \frac{1}{\Tr(\mm_i)} \mm_i^2 + \frac 2 K \id_n \preceq \frac{4\tau}{K}\id_n.
\end{align*}
The desired conclusion then follows by the matrix Bernstein bound (Theorem 1.6.2, \cite{Tropp-book}).
\end{proof}

Combining Propositions~\ref{prop:binary_search_lam},~\ref{prop:krylov}, and~\ref{prop:modified_boxspec} with Lemma~\ref{lem:lose_log} gives our main subroutine, which reweights a sum of matrices into a ``sparse'' component and a ``small'' component which we can recurse on.

\begin{proposition}\label{prop:sparse_plus_small}
Let $\eps \in (0, 1)$, $\delta \in (\exp(-\Omega(n)), 1)$ for an appropriate constant, and let $\{\mm_i\}_{i \in [m]} \subset \PSD^n$ satisfy $\allm(\1_m) \preceq \id_n$. There is an algorithm returning $u = v + w$ for $v, w \in \R^m_{\ge 0}$ satisfying 
\begin{equation}\label{eq:uvw}\normop{\allm(u - \1_m)} \le \eps,\;\nnz(w) \le \frac{\csparse n}{\eps^2},\; \normop{\allm(v)} \le \frac 1 {10},\end{equation}
with probability $\ge 1- \delta$, for a universal constant $\csparse$, in time
\[O\Par{\Par{\sum_{i \in [m]} \tmv(\mm_i)} \cdot \frac{\log(m) \log^2(\frac m \delta) \log\log^{11 + o(1)}(\frac m \delta)}{\eps^{3.5}}}.\]
\end{proposition}
\begin{proof}
Without loss of generality, assume $\eps$ is sufficiently small, and let $v_0 \defeq \0_m$. Lemma~\ref{lem:lose_log} with $\tau \gets n$ produces an initial reweighting $w_0$ of the matrices such that $\nnz(w_0) = O(\frac n {\eps^2} \log \frac n \delta)$, and $\normop{\allm(w_0 - \1_m)} \le \eps$. Our algorithm proceeds in $K = O(\log \log \frac n \delta)$ phases indexed by $k \in [K]$. We define $\beta \defeq \frac 1 {20K}$, and for all $k \in [K]$, we let $m_{k - 1} \defeq \nnz(w_{k - 1})$, and set the discrepancy radius 
\[\rho_k \defeq \max\Par{2\Cset \sqrt{\frac n {m_{k - 1}}}, \frac{\eps}{(K - k + 1)\log^2(K - k + 1) }}.\]
We terminate if some phase has $m_k \le \frac{\csparse n}{\eps^2}$. Otherwise, each phase $k$ computes $x_k$ such that for at least $ \frac{\ctight}{4} m_{k - 1}$ coordinates $i$, $[x_k]_i \le -1 + \beta$ and 
\begin{equation}\label{eq:xk_guarantee}\normop{\allm(w_{k - 1} \circ x_k)} = O(\rho_k).\end{equation}
We then let $S_k$ be the set of coordinates where $[x_k]_i \le -1 + \beta$, and update
\begin{gather*}
[w_k]_i \gets 0,\; [v_k]_i \gets [v_{k - 1}]_i + [w_{k - 1}]_i (1 + [x_k]_i) \text{ for all } i \in S_k, \\
[w_k]_i \gets (1 + [x_k]_i) [w_{k - 1}]_i \text{ for all } i \not\in S_k \text{ with } [w_{k - 1}]_i \neq 0.
\end{gather*}
Evidently the update rules imply that in every iteration, letting $u_k \defeq v_k + w_k$,
\begin{align*}
\normop{\allm(u_k - u_{k - 1})} &= \normop{\allm(w_{k - 1} \circ x_k)} = O(\rho_k), \\
\nnz(w_k) &\le \Par{1 - \frac{\ctight}{4}}\nnz(w_{k - 1}).
\end{align*}
Inducting on the second line above, $m_k$ is geometrically decreasing and at least $\Omega(\frac n {\eps^2})$ by the termination condition, so $\sum_{k \in [K]}\rho_k = O(\eps)$. Hence, telescoping the first line above using the triangle inequality shows that for each $k \in [K]$, we have $\normop{\allm(u_k)} \le \normop{\allm(u_0)} + O(\eps) \le 2$ for sufficiently small $\eps$. This in turn implies by the update rule on $v_k$ that
\[\normop{\allm(v_K)} \le \sum_{k \in [K]} \beta \normop{\allm(w_{k - 1})} \le \sum_{k \in [K]} \beta \normop{\allm(u_{k - 1})} \le \frac 1 {10}.\]
Hence, at termination we have all three guarantees in \eqref{eq:uvw}, after adjusting $\eps$ by a constant. 

To produce $x_k$ in each phase satisfying \eqref{eq:xk_guarantee}, we first draw $g \sim \Nor(\0_{m_{k - 1}}, \id_{m_{k - 1}})$, and define the matrices $\ma_i \gets [w_{k - 1}]_i \mm_i$ for all $i$ where $[w_{k - 1}]_i \neq 0$; note that under this definition, $\alla(\1_{m_{k - 1}}) \preceq 2\id_n$ as proven earlier. The precondition of Corollary~\ref{cor:gaussian_measure} is satisfied, so the conclusion of Proposition~\ref{prop:many_tight} holds except with exponentially small probability. We then apply Proposition~\ref{prop:binary_search_lam} (with Propositions~\ref{prop:krylov} and~\ref{prop:modified_boxspec} as subroutines) to solve both of the following problems (following notation \eqref{eq:set_r_def}) 
\[\min_{x \in \set_{\rho_k, \alla} \cap [-1, 1]^{m_{k - 1}}} \norm{x - g}_2^2,\; \min_{x \in \set_{\rho_k, \alla} \cap [-1, 1]^{m_{k - 1}}} \norm{x + g}_2^2,\]
to additive error $\tau = \frac{\ctight m_{k - 1}\beta^2}{4}$. In the context of Proposition~\ref{prop:binary_search_lam}, $\set = \set_{\rho_k, \alla}$ is defined with respect to the discrepancy radius $\rho_k$, $f(x) = \normop{\alla(x)}$, and $\alla$ is the operator associated with the $\{\ma\}_{i \in [m_{k - 1}]}$. Note that $f(x)$ satisfies \eqref{eq:rangedef} with $\Theta = 2$ by the assumption $\alla(\1_{m_{k - 1}}) \preceq 2\id_n$. Hence, both Assumption~\ref{assume:rs_big} and the conclusion of Proposition~\ref{prop:many_tight} both hold over all phases with probability at least $1 - \frac \delta 2$ by a union bound, so Lemmas~\ref{lem:many_nearly_tight} and~\ref{lem:sc_distsquare} show at least one of the approximate solutions (to $g$ or $-g$) has at least $\frac{\ctight m_{k - 1}}{4}$ coordinates in $[-1, -1 + \beta]$. The runtime follows from applying Propositions~\ref{prop:krylov},~\ref{prop:modified_boxspec} with failure probability $\frac \delta {4K}$, since each accuracy parameter in Proposition~\ref{prop:binary_search_lam} satisifes 
\[\frac {\lam \tau} 4 = \Omega(\rho_k \beta^2) = \Omega\Par{\frac{\eps}{\log\log^{3 + o(1)} \frac m \delta}},\]
by the second term in the definition of $\rho_k$ and our upper bound on $K$. We lose two additional $\log\log \frac m \delta$ factors: one from running $K$ phases, and one from the $T$ overhead in Proposition~\ref{prop:binary_search_lam}.
\end{proof}

By recursively applying Proposition~\ref{prop:sparse_plus_small}, we finally obtain the main result of this section.

\begin{theorem}\label{thm:fast_bss}
Let $\eps \in (0, 1)$, $\delta \in (\exp(-\Omega(n)), 1)$ for an appropriate constant, and let $\{\mm_i\}_{i \in [m]} \subset \PSD^n$ satisfy $\allm(\1_m) = \id_n$. There is an algorithm which returns $w \in \R^m_{\ge 0}$ satisfying $\nnz(w) = O(\frac{n}{\eps^2})$ and $\normop{\allm(w) - \id_n} \le \eps$, with probability $\ge 1 - \delta$, in time
\[O\Par{\Par{\sum_{i \in [m]} \tmv(\mm_i)} \cdot \frac{\log(m) \log^2(\frac m \delta) \log\log^{11 + o(1)}(\frac m \delta) }{\eps^{3.5}}}.\]
\end{theorem}
\begin{proof}
We proceed in $K = O(\log \frac 1 \eps)$ phases indexed by $k \in [K]$, and in each iteration we set
\[\eps_k = \eps \cdot (1 + c)^k,\]
where $K$ satisfies $10^K \ge \frac 1 \eps$ and $c$ is a small enough constant such that $\eps_K < 1$. We initialize $\bv_0 \gets \1_m$ and $\bw_0 \gets \0_m$. In each phase $k \in [K]$, we then run the algorithm in Proposition~\ref{prop:sparse_plus_small} on $\{ [\bv_{k - 1}]_i\mm_i\}_{i \in S_{k -1}}$ where $S_{k -1}$ is the set of $i$ with $[\bv_{k -1}]_i \neq 0$, with parameter $\eps_k$. Proposition~\ref{prop:sparse_plus_small} then produces $u_k = w_k + v_k$ satisfying the conclusions of Proposition~\ref{prop:sparse_plus_small}, i.e.\
\[\normop{\allm(\bv_{k - 1} - u_k)} \le \eps_k,\; \nnz(w_k) \le \frac{\csparse n}{\eps_k^2},\; \normop{\allm(v_k)} \le \frac 1 {10}.\]
Finally, we update $\bv_k \gets 10 v_k$ and $\bw_k \gets \bw_{k - 1} + \frac 1 {10^{k -1}} w_k$. We observe that for all $k \in [K]$, the above guarantees imply that $\normop{\allm(\bv_k)} \le 1$ and hence it was valid to run Proposition~\ref{prop:sparse_plus_small}. Moreover,
\begin{align*}
\bw_K &= \sum_{k \in [K]} (\bw_k - \bw_{k -1}) = \sum_{k \in [K]} \frac 1 {10^{k - 1}} (u_k - v_k) \\
&= \sum_{k \in [K]} \frac 1 {10^{k - 1}} \Par{ u_k - \bv_{k - 1}} +  \sum_{k \in [K]} \frac 1 {10^{k - 1}} \Par{\bv_{k - 1} - v_k} \\
&= \sum_{k \in [K]} \frac 1 {10^{k - 1}} \Par{ u_k - \bv_{k - 1}} + \bv_0 - \frac 1 {10^{K - 1}} v_K.
\end{align*}
Therefore, since $\bv_0 = \1_m$, we can rearrange the above and apply the triangle inequality to obtain
\begin{align*}
\normop{\allm(\bw_K) - \id_n} &\le \sum_{k \in [K]} \frac 1 {10^{k -1}} \normop{\allm(\bv_{k - 1} - u_k)} + \frac 1 {10^{K -1}} \normop{\allm(v_K)} \\
&\le \sum_{k \in [K]} \frac {\eps_k} {10^{k -1}} + \frac 1 {10^K} = O(\eps).
\end{align*}
Finally, the sparsity of $\bw_K$ is clearly $O(\frac n {\eps^2})$ since the sparsities of each $w_k$ are geometrically decreasing, and the conclusion follows by shifting $\eps$ by a constant and $\delta$ by a factor of $K$. We note that we do not pay a $K$-factor overhead in the runtime of Proposition~\ref{prop:sparse_plus_small} as $\eps_k$ is geometrically increasing.
\end{proof}

We remark that Theorem~\ref{thm:fast_bss} does not strongly rely on the assumption that $\allm(\1_m) = \id_n$: indeed, any initially $O(1)$-bounded operator norm matrix sum enjoys the same guarantees up to constant factors. However, to obtain a $\eps$-multiplicative spectral sparsification guarantee via Theorem~\ref{thm:fast_bss}, we require a constant spectral lower bound as well. Finally, the restriction $\delta \ge \exp(-\Omega(n))$ can be removed at a logarithmic overhead by running independent copies with constant failure probability.

\subsection{Graph sparsification}\label{ssec:graph_bss}

We now specialize the machinery in this section to the setting of graphs. Throughout this section let $G = (V, E, w_G)$ be a weighted graph, and let $m \defeq |E|$ and $n \defeq |G|$; we define the matrices $\mb$, $\mw$, and $\lap$ as in Section~\ref{sec:prelims}. We recall the following result on constructing an approximation to $\lap^\dagger$.

\begin{proposition}[\cite{KoutisMP11}]\label{prop:kmp}
Given a graph $G = (V, E, w_G)$ with $m = |E|$, $n = |V|$, and $\eps \in (0, 1)$, there is a randomized algorithm which returns a linear operator $\tlap: \R^V \to \R^V$ such that we can run the algorithm and apply $\tlap$ in time $O(m\log^{1 + o(1)}(n)\log \frac 1 \eps)$, and with high probability in $n$,
\[(1 - \eps)\lap^\dagger \preceq \tlap \preceq (1 + \eps)\lap^\dagger. \]
\end{proposition}

We now consider applying Theorem~\ref{thm:fast_bss} to the matrices
\[\mm_e \defeq \mw^{\half}\mb \tlap (w_e b_eb_e^\top) \tlap \mb^\top \mw^{\half},\text{ for all } e \in E,\]
where $\tlap$ is the matrix from Proposition~\ref{prop:kmp} with parameter $\eps$, and $w \defeq w_G$. Since
\[(1 - \eps)^2 \lap^\dagger \preceq \tlap \lap \tlap \preceq (1 + \eps)^2 \lap^\dagger, \]
and $\allm(\1_E) = \mw^{\half}\mb\tlap\lap\tlap\mb^\top\mw^{\half}$, defining $\proj_{\lap} \defeq \mw^{\half}\mb\lap^\dagger\mb^\top\mw^{\half}$, we have
\[(1 - \eps)^2 \proj_{\lap} \preceq \allm(\1_E) \preceq (1 + \eps)^2 \proj_{\lap}.\]
Notice that $\proj_{\lap}$ is the projection matrix onto the column span of $\mw^{\half}\mb$. It is straightforward to check that the argument in Section~\ref{ssec:main-bss} then returns a reweighting $x$ such that $\allm(x - \1_E) \preceq \eps\proj_{\lap}$, with $\nnz(x) = O(\frac n {\eps^2})$; this is made explicit by Theorem~\ref{thm:gen_rr20}, which only cares about the trace bound on $\allm(\1_E)$ (which is $n$), rather than the ambient dimension $m$. Finally, assuming $\eps$ is sufficiently small, we may rewrite this condition as
\begin{gather*} (1 - 4\eps)\proj_{\lap} \preceq \allm(w) \preceq (1 + 4\eps) \proj_{\lap} \\
\implies (1 - 4\eps)\lap \preceq\lap \tlap \Par{\sum_{e \in E} x_e w_e b_eb_e^\top} \tlap \lap \preceq (1 + 4\eps)\lap \\
\implies (1 - 7\eps)\lap \preceq \sum_{e \in E} x_e w_e b_eb_e^\top \preceq (1 + 7\eps)\lap.\end{gather*}
In the second line, we left-multiplied all matrices by $\mb^\top \mw^{\half}$ and right-multiplied by $\mw^{\half}\mb$, and in the third line, we first multiplied on both sides by $\lap^\dagger$, and then multiplied on both sides by $\tlap^\dagger$ (using that $\tlap$ and $\lap$ share a kernel) and used Proposition~\ref{prop:kmp}. Finally, the matrix-vector multiplication time to all $\{\mm_e\}_{e \in E}$ is dominated by applying $\tlap$. Adjusting $\eps$ by a constant, we obtain our claimed result on linear-sized sparsification, reproduced here for convenience.

\restatebssgraph* %
\section{Ultrasparsifiers}\label{sec:ultra}

In this section, we give a short application of the techniques in Section~\ref{sec:bss} to the setting of ultrasparsifiers. We will be concerned with the following problem: we are given a set of matrices $\{\mm_i\}_{i \in [m]} \subset \PSD^n$ and an additional matrix $\mn \in \PSD^n$, satisfying $\allm(\1_m) + \mn = \id_n$ where $\allm$ is the operator associated with the set in the sense of \eqref{eq:alla_def}. We also are promised
\[\Tr(\allm(\1_m)) \le \tau,\]
for a parameter $\tau$ such that $\frac n \tau$ is sufficiently large. We also require the relatively mild assumption that $n$ is subexponential in $m$; indeed, in our applications $m \ge n$. Our goal is to construct $w \in \R^m_{\ge 0}$ such that $\nnz(w) = O(\tau)$, and $\normop{\allm(w) + \mn - \id_n} \le \half$. We explain the connections of this problem to primitives in spectral graph theory at the end of the section, after giving our sparsification algorithm in Theorem~\ref{thm:ultra}. We begin by stating the following variant of Proposition~\ref{prop:sparse_plus_small}.

\begin{corollary}\label{cor:sparse_plus_small_ultra}
Let $\eps \in (0, 1)$, $\delta \in (\exp(-\Omega(\tau)), 1)$ for an appropriate constant, and let $\{\mm_i\}_{i \in [m]} \subset \PSD^n$ satisfy $\allm(\1_m) \preceq \id_n$ and $\Tr(\allm(\1_m)) \le \tau$, where $\frac n \tau$ is sufficiently large and $n = \exp(O(m))$ for an appropriate constant. There is an algorithm returning $u = v + w$ for $v, w \in \R^m_{\ge 0}$ satisfying
\[\normop{\allm(u - \1_m)} \le \eps,\; \nnz(w) \le \frac{\csparse\tau}{\eps^2},\; \normop{\allm(v)} \le \frac 1 {10}, \]
with probability $\ge 1 - \delta$, for a universal constant $\csparse$, in time
\[O\Par{\Par{\sum_{i \in [m]} \tmv(\mm_i)} \cdot \frac{\log(n)\log^2(\frac{mn}{\delta}) \log\log^{11 + o(1)}(\frac n \delta)}{\eps^{3.5}}}.\]
\end{corollary}
\begin{proof}
The proof is essentially identical to the proof of Proposition~\ref{prop:sparse_plus_small}, so we summarize the differences here, following the notation of our earlier proof. First, we initialize with the reweighting from Lemma~\ref{lem:lose_log} with a sparsity parameterized by $\tau$; there are still $K = \log \log \frac n \delta$ phases, each of which shrinks the sparsity of a maintained $w_k$ by a constant factor, until it reaches $\frac{\csparse\tau}{\eps^2}$. We use
\[\rho_k = \max\Par{2\Cset\sqrt{\frac{\tau}{m_{k - 1}}},\; \frac{\eps}{(K - k + 1)\log^2(K - k + 1)}},\]
and so for sufficiently large $\csparse$ we may continue using the measure lower bound from Corollary~\ref{cor:ultra_measure} at each phase with the framework from Section~\ref{sec:framework} to preserve the necessary invariants. The remaining parameters and runtime analysis of this algorithm are all the same as in Proposition~\ref{prop:sparse_plus_small}.
\end{proof}

\begin{theorem}\label{thm:ultra}
Let $\delta \in (\exp(-\Omega(\tau)), 1)$ for an appropriate constant, and let $\{\mm_i\}_{i \in [m]} \subset \PSD^n$ and $\mn \in \PSD^n$ satisfy $\allm(\1_m) + \mn = \id_n$ and $\Tr(\allm(\1_m)) \le \tau$, where $\frac n \tau$ is sufficiently large and $n = \exp(O(m))$ for an appropriate constant. There is an algorithm which returns $w \in \R^m_{\ge 0}$ satisfying $\nnz(w) = O(\tau)$ and $\normop{\allm(w) + \mn - \id_n} \le \half$, with probability $\ge 1 - \delta$, in time
\[O\Par{\Par{\sum_{i \in [m]}\tmv(\mm_i)} \cdot \log(n)\log^2\Par{\frac n \delta} \log\log^{11 + o(1)}\Par{\frac n \delta}}.\]
\end{theorem}
\begin{proof}
We recursively apply Corollary~\ref{cor:sparse_plus_small_ultra} in the same way Proposition~\ref{prop:sparse_plus_small} was used in Theorem~\ref{thm:fast_bss}. This yields a $O(\tau)$-sparse reweighting $w$ satisfying $\normop{\allm(w - \1_m)} = \normop{\allm(w) + \mn - \id_n } \le \half$, where it suffices to take $\eps \le \half$ to be a sufficiently small constant in the proof of Theorem~\ref{thm:fast_bss}. 
\end{proof}

We now describe consequences of Theorem~\ref{thm:ultra}. As discussed in Section~\ref{sec:intro}, the state-of-the-art ultrasparsifiers constructible in near-linear time have parameters $(\bO(\ell\log^2(n)), \ell)$, but there are higher-quality $(\min(\bO(\ell \log(n)), \ell^{1 + o(1)}\log^{o(1)}(n)), \ell)$-ultrasparsifiers constructible in polynomial time \cite{KollaMST10, JambulapatiS21}. We show how to match \cite{KollaMST10, JambulapatiS21} in near-linear time. Our starting point (as is the starting point of prior constructions \cite{SpielmanT04, KollaMST10, KoutisMP11, JambulapatiS21}) is the following notion of a low-distortion subgraph (generalizing the combinatorial notion of a low-stretch spanning tree, as is demonstrated formally by Theorem 2.1 of \cite{SpielmanW09}).

\begin{definition}\label{def:low-stretch}
We say subgraph $H$ is a $\sigma$-distortion subgraph of $G = (V, E, w_G)$ if $\Tr(\lap_H^{\dagger}\lap_G) \le \sigma$.
\end{definition}

We further recall the following known constructions of low-distortion subgraphs. The first result was used in \cite{KollaMST10} and the latter was both developed and used in \cite{JambulapatiS21}.

\begin{proposition}[\cite{AbrahamN19}]\label{prop:low-stretch-tree}
Let $G = (V, E, w_G)$ have $m = |E|$ and $n = |V|$. There is an algorithm running in time $O(m\log n \log\log n)$ which produces a $O(m\log n \log\log n)$-distortion subgraph of $G$. Further, this subgraph is a spanning tree of $G$.
\end{proposition}

\begin{restatable}[\cite{JambulapatiS21}]{proposition}{restatejssubgraph}\label{prop:js21_subgraph}
Let $G = (V, E, w_G)$ have $m = |E|$ and $n = |V|$. There is an algorithm which takes a parameter $\gamma > 1$ and runs in time $O(m)$, which produces a $O(m\gamma^{o(1)}(\log n)^{o(1)})$-distortion subgraph of $G$ with $n - 1 + \frac m \gamma$ edges.
\end{restatable}

Because the statement of Proposition~\ref{prop:js21_subgraph} is not explicit in \cite{JambulapatiS21}, we discuss how to obtain it in Appendix~\ref{app:subgraph_construct}. We now combine Propositions~\ref{prop:low-stretch-tree} and~\ref{prop:js21_subgraph} with Theorem~\ref{thm:ultra}. First, by applying Theorem~\ref{thm:bss_graph} with constant $\eps$, we may assume that $m = O(n)$ at the cost of constant spectral distortion and a near-linear runtime overhead. Overloading $G$ to mean the graph after applying this preprocessing, we let $H$ be a $\sigma$-distortion subgraph of $G$, and for a parameter $\kappa \ge 1$, we let
\begin{align*}
\mn &\defeq \Par{\kappa\lap_H + \lap_{G}}^{\frac \dagger 2} \Par{\kappa\lap_H}\Par{\kappa\lap_H + \lap_{G}}^{\frac \dagger 2}, \\
\mm_e &\defeq \Par{\kappa\lap_H + \lap_{G}}^{\frac \dagger 2} (w_e b_e b_e^\top)\Par{\kappa\lap_H + \lap_G}^{\frac \dagger 2} \text{ for all } e \in E.
\end{align*}
It is straightforward to see that $\mn + \sum_{e \in E} \mm_e = \idnoone$. Further, the distortion definition guarantees
\[\Tr\Par{\sum_{e \in E} \mm_e} = \Tr\Par{\Par{\kappa\lap_H + \lap_G}^\dagger \lap_G} \le \Tr\Par{\Par{\kappa\lap_H}^\dagger \lap_G} \le \tau \defeq \frac{\sigma}{\kappa}.\]
Finally, we note that (spectral approximations of) the square roots of Laplacian pseudoinverses may be applied in nearly-linear time, using the same strategy as in Section~\ref{ssec:graph_bss}. 
Therefore, Theorem~\ref{thm:ultra} (with our other tools) shows that we can efficiently obtain a graph on $n - 1 + O(\frac \sigma \kappa)$ edges whose Laplacian is an $O(1)$-spectral approximation to $\kappa\lap_H + \lap_G$, which itself is a $O(\kappa)$-spectral approximation to $\lap_G$. By taking $H$ as in Proposition~\ref{prop:low-stretch-tree} or~\ref{prop:js21_subgraph} and $\kappa$ large enough, and reparameterizing the problem in terms of $\ell$ (the edge reduction factor), we achieve a $\min(\bO(\ell\log(n)), \ell^{1 + o(1)}\log^{o(1)}(n))$-spectral approximation to $\lap_G$ on $n - 1 + \frac n \ell$ edges in nearly-linear time. We record this construction formally in Corollary~\ref{cor:ultra_fast}, restated for convenience: the runtime comes from combining Theorem~\ref{thm:ultra} with Section~\ref{ssec:graph_bss}, as the runtimes of Propositions~\ref{prop:low-stretch-tree},~\ref{prop:js21_subgraph} do not dominate.

\restateultrafast*

In particular, Corollary~\ref{cor:ultra_fast} yields the main ultrasparsifier results of \cite{KollaMST10, JambulapatiS21} with an alternative proof (without going through the \cite{BatsonSS14} analysis), and much faster algorithms.

 %
\section{Degree-preserving sparsifiers}\label{sec:degree}

In this section, we consider a degree-preserving variant of the specialization of Section~\ref{sec:bss} to graphs. For an undirected graph $G = (V, E, w_G)$ and a parameter $\eps \in (0, 1)$ (both fixed throughout), we define associated matrices and vectors $\lap$, $|\mb|$, $\mb$, $\{b_e\}_{e \in E}$ and $\mw$ as in Section~\ref{sec:prelims}. Our goal is to obtain a sparse
$w \in \R^E_{\ge 0}$ with 
\[(1 - \eps)\lap \preceq \sum_{e \in E} [w \circ w_G]_e b_e b_e^\top \preceq (1 + \eps)\lap \text{ and } |\mb| \mw (w - \1_E) = 0,\]
i.e.\ the degrees of the reweighted graph are unchanged from $G$. In Section~\ref{ssec:framework-degree} and Section~\ref{ssec:linear-oracle} we develop our main optimization subroutines for implementing the sparsification framework of \cite{RR20} with linear constraints. We then give a degree-preserving rounding procedure in Section~\ref{ssec:spectral_round}, and prove our main result Theorem~\ref{thm:main_degree} in Section~\ref{ssec:main-degree}. Finally, in Section~\ref{ssec:tradeoffs} we show how to slightly improve the runtime of Theorem~\ref{thm:main_degree} at the cost of a logarithmic overhead in the sparsity.

\subsection{Approximate Gaussian rounding with linear constraints}\label{ssec:framework-degree}

In this section, we provide a subroutine for optimizing the subproblems required by Proposition~\ref{prop:binary_search_lam}, under additional linear constraints denoted $\mc \in \R^{d \times m}$. We assume $d \le (1 - 2\ctight) m$ (where $\ctight$ is the constant from Proposition~\ref{prop:many_tight}). Consideration of this setting is motivated by the tolerance of Proposition~\ref{prop:many_tight} to constraining a sufficiently small number of linear directions. In applications, $\mc$ is a reweighted signed edge-vertex incidence matrix of a subgraph of $G$, with $d \le n \ll m$. 

As in Section~\ref{ssec:fast_gaussian_rounding}, we assume we are given $\{\ma_i\}_{i \in [m]} \subseteq \PSD^n$ with associated operator $\alla$, a vector $g \in \R^m$ satisfying Assumption~\ref{assume:rs_big}, and a parameter $\rho$ fixed throughout. We also explicitly assume $\normop{\alla(\1_m)} \le 2$, and define the discrepancy body $\set_{\rho, \alla}$ as in \eqref{eq:set_r_def} with respect to the given matrices. Our goal in this section will be to approximate the value $\rsc$ and point $\xsc$ defined as
\begin{equation}\label{eq:rscdef}\rsc \defeq \norm{\xs - g}_2, \text{ where } \xsc \defeq \arg\min_{\substack{x \in [-1, 1]^m \cap\set_{\rho, \alla} \\ \mc x = \0_d} }\norm{x - g}_2.\end{equation}
We recall that $f_\lam$ is defined in \eqref{eq:test_opt} (where $f(x) = \normop{\alla(x)}$ in this section), and we define
\begin{equation}\label{eq:xsetcdef}
\xsetc \defeq \{x \in [-1, 1]^m \mid \mc x = \0_d\}.
\end{equation}
In the remainder of this section we will discuss approximately solving $f_\lam$ over $\xsetc$. Our algorithm will use a variant of a Frank-Wolfe method, prompting the following oracle definition.

\begin{definition}[Linear optimization oracle]\label{def:linopt}
Let $\xset \subseteq \R^m$ be a symmetric convex set, and let $\Delta \in (0, 1)$. We say $\olinopt$ is a $\Delta$-approximate linear optimization oracle over $\xset$ if on input $g \in \R^m$ it returns $z \in \xset$ satisfying $c^\top z \le (1 - \Delta) \min_{x \in \xset} c^\top x$. 
\end{definition}
We note that because $\xset$ is symmetric, Definition~\ref{def:linopt} is sensible as the minimum is always nonpositive. Before giving our optimization procedure for $f_\lam$, we give a helper fact, where we collect several properties of a well-studied smooth approximation to the operator norm. 

\begin{fact}[\cite{Nesterov07}]\label{fact:logtrexp}
Let $\mu > 0$, let $\talla$ to be the operator associated with $\{\tma_i\}_{i \in [m]} \in \R^{2n \times 2n}$ as defined in \eqref{eq:test_opt}, and let $A_\mu(x) \defeq \mu\log \Tr \exp(\frac 1 \mu \talla(x))$. Then $\normop{\alla(x)} \le A_\mu(x) \le \normop{\alla(x)} + \mu \log(2n)$ for all $x \in \R^m$, and $A_\mu$ is $\frac 1 \mu$-smooth in the $\ell_\infty$ norm, i.e.\ for all $x, x' \in \R^m$,
\[A_\mu(x') \le A_\mu(x) + \inprod{\nabla A_\mu(x)}{x' - x} + \frac{1}{2\mu}\norm{x' - x}_\infty^2.\]
Further, $\nabla A_\mu(x) = \talla^*(\my(x))$, where $\my(x) = \exp(\frac 1 \mu \talla(x)) \cdot (\Tr\exp(\frac 1 \mu \talla(x)))^{-1}$.
\end{fact}

We also define the notion of an approximate gradient oracle, and recall a (standard) efficient construction of such an oracle from prior work based on approximations to the exponential.

\begin{definition}[Approximate gradient oracle]\label{def:approx_grad}
Following the notation of Fact~\ref{fact:logtrexp}, we say $g: [-1, 1]^m \to \R^m$ is a $\Delta$-approximate gradient oracle for $A_\mu$ if $\norm{g(x) - \nabla A_\mu(x)}_1 \le \Delta$ for all $x \in [-1, 1]^m$.
\end{definition}

The following claim applies Proposition 2 of \cite{JambulapatiT23} with $\eps \gets O(\Delta)$ and $\gamma \gets O(\frac \Delta n)$.

\begin{lemma}[Proposition 2, \cite{JambulapatiT23}]\label{lem:approx_grad}
We can implement a $\Delta$-approximate gradient oracle for $A_\mu$ which succeeds with probability $\ge 1 - \delta$ in time
\[O\Par{\Par{\sum_{i \in [m]} \tmv(\ma_i)} \cdot \frac{\log^2\Par{\frac{mn}{\mu\Delta\delta}}}{\sqrt{\mu}\Delta^2}}\]
\end{lemma}

We are now ready to give our main Frank-Wolfe subroutine.

\begin{lemma}\label{lem:minflam_constraint}
Let $\beta, \rho, \delta \in (0, 1)$, $\lam \ge \frac{\rho}{8m}$, and let $\olinopt$ be a $\Delta$-approximate linear optimization oracle over $\xsetc$ for $\Delta \defeq \frac{\ctight \beta^2}{240}$. There is an algorithm which solves $f_\lam$ to additive error $\frac{\lam \ctight m \beta^2}{16}$ over $\xsetc$ using $N = O(\frac{\log n}{\rho^2\beta^4})$ calls to $\olinopt$ and 
\[O\Par{\Par{\sum_{i \in [m]} \tmv(\ma_i)} \cdot \Par{\frac{\log^{3.5}\Par{\frac{mn}{\rho\beta\delta}}}{\rho^{4.5}\beta^9}}}\]
additional time, with probability $\ge 1 - \delta$.
\end{lemma}
\begin{proof}
Our algorithm will be a Frank-Wolfe method applied to the smooth approximation from Fact~\ref{fact:logtrexp}. Throughout the proof, define for convenience $x^\star$ to be the minimizer of $f_\lam$ over $\xsetc$, and
\[\mu \defeq \frac{\lam \ctight m\beta^2}{48\log(2n)},\; F(x) \defeq A_\mu(x) + \lam \norm{x - g}_2^2.\]
By the fact that $\norm{\cdot}_2^2 \le m \norm{\cdot}_\infty^2$, we have that $F$ is $L$-smooth in the $\ell_\infty$ norm for
\[L \defeq \frac{1}{\mu} + 2\lam m.\]
Finally, we let $g$ be a $\Delta'$-approximate gradient oracle for $A_\mu$, for
\[\Delta' \defeq \frac{L}{N + 1}.\]
Our algorithm defines $\eta_t \defeq \frac 2 {t + 1}$ for all $t \in [N]$ and initializes at $x_1 = \0_m$. It then repeatedly updates $x_{t + 1} \gets (1 - \eta_t) x_t + \eta_t z_t$ where $z_t$ is the result of calling $\olinopt$ on $g(x_t) + 2\lam(x_t - g)$. In each iteration, since $g(x_t) + 2\lam(x_t - g)$ approximates $\nabla F(x_t)$ up to $\Delta'$ in the $\ell_1$ norm,
\begin{align*}
F(x_{t + 1}) &\le F(x_t) + \eta_t\inprod{\nabla F(x_t)}{z_t - x_t} + \frac{L\eta_t^2}{2}\norm{z_t - x_t}_\infty^2 + 2\eta_t \Delta' \\
&\le F(x_t) + \eta_t\inprod{\nabla F(x_t)}{(1 - \Delta)x^\star - x_t} + 3L\eta_t^2,
\end{align*}
since $\norm{z_t - x_t}_\infty \le 2$ for $z_t, x_t \in \xsetc$, and we applied Definition~\ref{def:linopt} with $\xs \in \xsetc$. Now defining $E_t \defeq F(x_t) - F((1 - \Delta)\xs)$, the above display and convexity of $F$ implies
\begin{equation}\label{eq:error_recursion}
E_{t + 1} \le E_t + \eta_t\inprod{\nabla F(x_t)}{(1 - \Delta)x^\star - x_t} + 2L\eta_t^2 \le (1 - \eta_t)E_t + 3L\eta_t^2.
\end{equation}
We next claim that inductively, $E_t \le \frac{12L}{t + 1}$ for all $t \in [N]$. The base case of the induction follows from nonnegativity of $F$ everywhere, so Fact~\ref{fact:logtrexp} and Assumption~\ref{assume:rs_big} imply
\[E_1 \le F(\0_m) \le \mu \log(2n) + 4\lam m \le 4L.\]
Further, inductively applying \eqref{eq:error_recursion} preserves this invariant:
\begin{align*}
E_{t + 1} \le \frac{t - 1}{t + 1} \cdot \frac{12L}{t + 1} + \frac{12L}{(t + 1)^2} < \frac{12L}{t + 2}.
\end{align*}
By taking the constant in front of $N$ large enough, and using $\lam \ge \frac{\rho}{8m}$, this implies
\begin{align*}
f_\lam(x_N) - f_\lam(\xs) &\le F(x_N) - F(\xs) + \mu\log(2n) \\
&\le F(x_N) - F((1 - \Delta)\xs) + \Delta F(\0_m) + \mu\log(2n) \\ 
&\le \frac{\lam \ctight m\beta^2}{48} + \Delta F(\0_m) + \frac{\lam \ctight m\beta^2}{48} \le \frac{\lam \ctight m\beta^2}{16}.
\end{align*}
where the first inequality used Fact~\ref{fact:logtrexp}, the second used nonnegativity and convexity of $F$, the third used our bound on $E_N$, and the last used $F(\0_m) \le \mu\log(2n) + 4\lam m \le 5\lam m$. The runtime comes from applying Lemma~\ref{lem:approx_grad} $N$ times, with $\delta \gets \frac \delta N$, and union bounding for the failure probability. Here, we note that it suffices to take $\Delta' = \Theta(\rho\beta^2)$ and $\frac 1 \mu = O(\frac{\log n}{\rho\beta^2})$.
\end{proof}

To complete our use of Lemma~\ref{lem:minflam_constraint}, we will give an efficient linear optimization oracle in Section~\ref{ssec:linear-oracle} when $\xset = \xsetc$ and the constraint matrix $\mc$ is graph-structured.

\subsection{Degree-preserving linear optimization oracle}\label{ssec:linear-oracle}

We recall the following standard facts from prior work (e.g.\ \cite{BallCL94}) used in our analysis.

\begin{fact}\label{fact:pnorm}
For $p \ge 2$, $\half \norm{\cdot}_p^2$ is $(p-1)$-smooth in the $\ell_p$ norm where the domain is $\R^d$ for any $d \in \N$. Additionally, for any $x \in \R^d$ and $p \ge q \ge 1$, we have $\norm{x}_p \le \norm{x}_q \le d^{\frac 1 q - \frac 1 p} \norm{x}_p$.
\end{fact}

Further, in this section we will require the notion of an oblivious routing.

\begin{definition}\label{def:route}
We say $\mr \in \R^{E \times V}$ is an $\alpha$-oblivious routing for a graph $G = (V, E, w_G)$ if $\mb^\top \mr x = x$ for all $x \in \R^V$, and $\norm{\mw^{-1}\mr \mb^\top \mw y}_\infty \le \alpha \norm{y}_\infty$ for all $y \in \R^E$. We let $\troute(\alpha)$ be the time required to construct and apply an $\alpha$-oblivious routing when $G$ is clear from context.\end{definition}

We refer the reader to \cite{KelnerLOS14} for more exposition on the history and uses of oblivious routings. We also recall the following construction of an oblivious routing from that paper.

\begin{proposition}[Theorem 19, \cite{KelnerLOS14}]\label{prop:klos}
Let $G = (V, E, w_G)$ have $m = |E|$ and $n = |V|$. Then $\troute(\alpha) = mn^{o(1)}$ for some $\alpha = n^{o(1)}$, with high probability in $n$.
\end{proposition}

Finally, we give two simple helper claims used in our proof.

\begin{lemma}\label{lem:relate_opt}
For some $p \ge 1$, $v \in \R^{m}$, and $\ma \in \R^{m \times m}$, let $\opt_{v, p} \defeq \min_{\norm{\ma z}_p \le 1} \inprod{v}{\ma z}$. Then
\[\min_{z \in \R^m} \inprod{v}{\ma z}+ \half \norm{\ma z}_p^2 = -\half \opt_{v, p}^2.\]
Further, for any $\Delta \in (0, 1)$ and $w \in \R^m$ with
\[\min_{w \in \R^m} \inprod{v}{\ma w} + \half\norm{\ma w}_p^2 \le -\frac{(1 - \Delta)^2}{2}\opt_{v, p}^2,\]
we can obtain $w'$, a rescaling of $w$, satisfying $\norm{\ma w'}_p \le 1$ and $\inprod{v}{\ma w'} \le (1 - \Delta)\opt_{v, p}$.
\end{lemma}
\begin{proof}
Letting $z^\star$ achieve value $\opt_{v, p}\le 0$, by setting $z \gets \lam z^\star$ for $\lam = -\opt_{v, p}$, we see that $z$ attains value $-\half \opt_{v, p}^2$ in the above display; here we use $\norm{\ma z^\star}_p = 1$, as otherwise we could scale up $z^\star$. To see that no $z$ attains smaller value, suppose otherwise, and let $z$ minimize the above display achieving value $-\half C^2 < -\half\opt_{v, p}^2$. Since rescaling $z$ by any $\lam$ does not improve the value,
\[\frac \dd {\dd\lam}\Par{\lam\inprod{v}{\ma z} + \frac{\lam^2}{2}\norm{\ma z}_p^2} = 0\text{ at } \lam = 1 \iff \inprod{v}{\ma z} = -\norm{\ma z}_p^2 = -C^2.\]
Therefore, $\frac 1 C\norm{\ma z}_\infty \le 1$, and $\frac z C$ achieves value $-C < \opt_{v, p}$ for the original problem, a contradiction. To see the second claim, suppose without loss of generality we have first rescaled so that $\inprod{v}{\ma w} = -\norm{\ma w}_p^2 \le -(1 - \Delta)^2\opt_{v, p}^2$, as otherwise rescaling by a constant $\lam$ improves the function value. We then rescale $w' \gets \frac w {\norms{\ma w}_p}$, and the conclusion follows immediately.
\end{proof}

\begin{lemma}\label{lem:bounded_radius}
Let $p \ge 1$, and suppose
\[\inprod{v}{\ma z} + \half\norm{\ma z}_p^2 \le 0,\]
where $\opt_{v, p}$ is defined as in Lemma~\ref{lem:relate_opt}. Then $\norm{\ma z}_\infty \le -2\opt_{v, p}$.
\end{lemma}
\begin{proof}
First, we claim any $z$ with $\norm{\ma z}_p > -2\opt_{v, p}$ has $\inprod{v}{\ma z} + \half\norm{\ma z}_p^2 > 0$. To see this,
\begin{align*}
\inprod{v}{\ma z} + \half\norm{\ma z}_p^2 &\ge \norm{\ma z}_p \Par{\inprod{v}{\ma \frac{z}{\norm{\ma z}_p}}} + \half\norm{\ma z}_p^2 \\
&\ge \norm{\ma z}_p\Par{\opt_{v, p} + \half\norm{\ma z}_p} > 0,
\end{align*}
where we used the definition of $\opt_{v, p}$. The conclusion follows from $\norm{\ma z}_\infty \le \norm{\ma z}_p$.
\end{proof}

We are now ready to give our linear optimization oracle.

\begin{lemma}\label{lem:linear_opt_oracle}
We can implement a $\Delta$-approximate linear optimization oracle over $\xset^{\mb^\top\mw}$, where $\mb$, $\mw$ are associated with $G = (V, E, w_G)$, in time $O(\troute(\alpha) \cdot \frac{\alpha^2\log |E|}{\Delta^2})$ for any $\alpha \ge 1$.
\end{lemma}
\begin{proof}
Throughout we let $\mr$ be an $\alpha$-oblivious routing, identify $V$ and $E$ with $[n]$ and $[m]$ in the standard way for $n = |V|$ and $m = |E|$, and define $\mc \defeq \id_m - \mr \mb^\top$ and $\ma \defeq \mw^{-1}\mc \mw$. By the definition of an oblivious routing, $\mb^\top \mc \mw z = \0_n$ for all $z \in \R^m$. Hence, we may solve
\[\min_{\norm{\ma z}_\infty \le 1} \inprod{c}{\ma z}\]
to $\Delta$-multiplicative accuracy, and output $x = \ma z$, which satisfies 
\[\mb^\top \mw x = \mb^\top \mc \mw z = \0_n.\]
Next, note that for $p \defeq \frac{5\log m}{\Delta}$, Fact~\ref{fact:pnorm} implies
\[\Par{1 - \frac \Delta 4}\opt_{c, \infty} > \opt_{c, p} > \opt_{c, \infty},\]
so Lemma~\ref{lem:relate_opt} implies it suffices to return a solution to
\begin{equation}\label{eq:pnorm_obj}\min \inprod{c}{\ma z} + \half\norm{\ma z}_p^2\end{equation}
up to $\frac \Delta 4 \opt_{c, p}^2$ additive error. Next, the oblivious routing definition and triangle inequality imply
\[\max_{\norm{v}_\infty = 1}\norm{\ma v}_\infty \le 1 + \alpha \le 2\alpha,\]
so Fact~\ref{fact:pnorm} shows that the objective in \eqref{eq:pnorm_obj} is $4\alpha^2 p$-smooth in the $\ell_\infty$ norm. Our linear optimization oracle iteratively takes a step of the $\ell_\infty$ gradient descent algorithm in Theorem 1 of \cite{KelnerLOS14} (on the objective \eqref{eq:pnorm_obj}), and multiplies the iterate by $\ma$. The objective value for iterates of $\ell_\infty$ gradient descent is monotone non-increasing, and multiplication by $\ma$ preserves the objective value since it is simple to verify $\ma^2 z = \ma z$ for all $z$. Hence, Lemma~\ref{lem:bounded_radius} implies that the $\ell_\infty$ radius of all iterates is bounded by $O(\opt_{v, p})$. A direct application of Theorem 1 in \cite{KelnerLOS14} finally implies the algorithm returns a point with function error $\frac \Delta 4 \opt_{c, p}^2$ to \eqref{eq:pnorm_obj} in
\[k = O\Par{\frac{\alpha^2 p \opt_{c, p}^2}{\Delta \opt_{c, p}^2}} = O\Par{\frac{\alpha^2 \log m}{\Delta^2}}\]
iterations, each of which is dominated by multiplication through $\ma$, giving the runtime.
\end{proof}

\subsection{Degree-preserving rounding via link/cut trees}\label{ssec:spectral_round}

In this section, we give a self-contained subroutine solving the following degree-preserving rounding problem via an application of the link/cut tree dynamic data structure of \cite{SleatorT83}.

\begin{lemma}\label{lem:degree-preserve-round}
Let $\mb$, $\mw$ be associated with $G = (V, E, w_G)$, define $u_e \defeq 2[w_G]_e$ for all $e \in E$, and let $n = |V|$, $m = |E|$. There is an algorithm which runs in time $O(m \log n)$ and outputs $w \in \prod_{e \in E} [0, u_e]$ such that $|\mb|^\top w = |\mb|^\top w_G$, and $w_e = 0$ for at least $\frac{m - (n - 1)}{2}$ of the $e \in E$.
\end{lemma}
\begin{proof}
The algorithm is very similar to the cycle-cancelling data structure via link/cut trees in Appendix C of \cite{AssadiJJST22}, so we first give a high-level description of the algorithm, and then discuss how to modify Appendix C of \cite{AssadiJJST22} to implement the algorithm in $O(m \log n)$ time. At the beginning of the algorithm, all edges in $E$ are marked ``alive,'' and we set $w \gets w_G$. We then iteratively find a cycle $C$ in the graph restricted to alive edges, and let 
\[\Delta \defeq \min_{e \in C} \min\Par{u_e - w_e, w_e}.\]
In other words, $\Delta$ is the closest distance of any current weight to saturating either the lower bound $0$ or the upper bound $u_e$. We then alternately add and subtract $\Delta$ from $w$ along the cycle so that one edge is saturated (i.e.\ either $0$ or $u_e$), and mark the saturated edge as ``dead.'' Inductively, $w \in \prod_{e \in E}[0, u_e]$ at all times, and $|\mb|^\top w = |\mb|^\top w_G$ is preserved. Further, at the end of the algorithm at most $n - 1$ edges are not saturated. If more saturated edges are set to $0$, then we output $w$. Otherwise, we output $2w_G - w \in \prod_{e \in E}[0, u_e]$, which preserves the degrees as
\[|\mb|^\top(2w_G - w) = 2|\mb|^\top w_G - |\mb|^\top w = |\mb|^\top w_G.\]
We now discuss how to implement this algorithm efficiently, following the notation of Appendix C of \cite{AssadiJJST22}. We will maintain four link/cut forests with the same topology, denoted $T^+_{\up}$, $T^+_{\lo}$, $T^-_{\up}$, and $T^-_{\lo}$, and process the edges of $G$ sequentially. When we process an edge $e$, we add it to all four forests. The role of the forests $T^-_{\up}$ and $T^-_{\lo}$ is to maintain the distances of $w_e$ from $u_e$ and $0$ respectively, for all $e$ with odd depth; similarly, $T^+_{\up}$ and $T^+_{\lo}$ maintain distances of the weights of all even-depth edges from saturation. When we process an edge $e = (u, v)$, if its depth in the tree containing $u$ is odd, we set its weight in $T^-_{\up}$ and $T^-_{\lo}$ to $[w_G]_e$, and set its weight in $T^+_{\up}$ and $T^+_{\lo}$ to $2m\norm{w_G}_\infty$; we symmetrically handle even-depth edges. Next, if the endpoints of $e$ belong to different trees in the forest, we link these trees. Otherwise, we have found a cycle, and we can find the minimum distance to saturation $\Delta$ along the path between $u$ and $v$ in $O(\log n)$ time. By the weight setting of $2m\norm{w_G}_\infty$ on odd-depth edges in $T^+_{\up}$ and $T^+_{\lo}$ at initialization, these edges will never minimize the weight on any path throughout the algorithm, since they stay at least $2\norm{w_G}_\infty$; even-depth edges are handled identically. We then add or subtract $\Delta$ from the path in all four trees appropriately in $O(\log n)$ time, and delete the saturated edge; hence, the runtime is $O(m \log n)$. 
\end{proof}

\subsection{Degree-preserving sparsifiers in almost-linear time}\label{ssec:main-degree}

We now prove Theorem~\ref{thm:main_degree} after stating one helper lemma.

\begin{lemma}\label{lem:bipartite_b}
Let $G = (V, E, w_G)$ be bipartite with bipartition $V = L \cup R$, and let $\mb \in \R^{E \times V}$ be its edge-incidence matrix signed such that for each $e = (u, v) \in E$ for $u \in L$ and $v \in R$, $\mb_{e:}$ has a $1$ in the $u$ coordinate and a $-1$ in the $v$ coordinate. Then  $\mb^\top x = 0$ implies $|\mb|^\top x = 0$.
\end{lemma}
\begin{proof}
It suffices to note that $|\mb|^\top = \md \mb^\top$, where $\md$ is a diagonal matrix with 
\begin{align*}
\md_{vv} = \begin{cases}
1 & v \in L \\
-1 & v \in R
\end{cases}.
\end{align*}
\end{proof}

\restatemaindegree*

\begin{proof}
Our algorithm follows analogously to the proof of Proposition~\ref{prop:sparse_plus_small} and its specialization in Section~\ref{ssec:graph_bss}, so we follow their notation, and identify $V$ and $E$ with $[n]$ and $[m]$ in the standard way. We initialize $w_0 \gets \1_m$, and define
\[\mm_e \defeq \mw^{\half}\mb\tlap ([w_G]_e b_eb_e^\top)\tlap\mb^\top\mw^{\half}\text{ for all } e \in E,\]
where we follow the notation of Section~\ref{ssec:graph_bss} in defining $\tlap$. Our algorithm proceeds in $K = O(\log m)$ phases indexed by $k \in [K]$. For all $k \in [K]$, we let $m_{k - 1} \defeq \nnz(w_{k - 1})$, and
\[\rho_k \defeq \max\Par{2\Cset\sqrt{\frac n {m_{k - 1}}}, \frac{\eps}{(K - k + 1)\log^2(K - k + 1)}}.\]
We again terminate if $m_k \le \frac{\csparse n}{\eps^2}$ for a universal $\csparse$. Otherwise, each phase $k$ computes $x_k$ such that $|\mb|^\top (w_{k - 1} \circ x_k) = \0_n$, and for at least $\frac{\ctight}{30} m_{k - 1}$ coordinates $i$, $[x_k]_i = -1$, and 
\[\normop{\allm(w_{k - 1} \circ x_k)} = O(\rho_k).\]
We then update $w_k \gets w_{k - 1} \circ (\1_m + x_k)$, and the remainder of the proof follows analogously to Proposition~\ref{prop:sparse_plus_small} and Section~\ref{ssec:graph_bss}. It remains to argue how to compute such a vector $x_k$. 

First, we randomly partition the vertices in the support of $w_{k - 1}$ into two sets, and accept if the number of edges crossing the bipartition is at least $\frac {m_{k - 1}} 3$; this occurs with high probability in $n$ after $O(\log n)$ random bipartitions, since by linearity of expectation and Markov's inequality, the number of edges not in the bipartition is at most $\frac {2m_{k - 1}} 3$ with constant probability. Let $G_{k} = (V, E_{k}, w_{k - 1} \circ w_G)$ where $E_k$ are the edges crossing the bipartition, let $\mb_{k}$ be its signed edge-vertex incidence matrix, and let $\ma_e = [w_{k - 1}]_e \mm_e$ for all $e \in E_k$. By Lemma~\ref{lem:bipartite_b}, the condition $|\mb_k|^\top([w_{k - 1} \circ w_G \circ x_k]_{E_{k}}) = \0_n$ is implied by $\mb_k^\top([w_{k - 1} \circ w_G \circ x_k]_{E_{k}}) = \0_n$. Hence, letting $\mw_k = \diag{[w_{k - 1} \circ w_G]_{E_k}}$, we draw $g \sim \Nor(\0_{m_{k - 1}}, \id_{m_{k - 1}})$, and solve both of the problems
\begin{align*}
\min_{x \in \set_{\rho_k, \alla} \cap \xset^{\mb_k^\top \mw_k}} \norm{x - g}_2^2,\; \min_{x \in \set_{\rho_k, \alla} \cap \xset^{\mb_k^\top \mw_k}} \norm{x + g}_2^2,
\end{align*}
to additive error $\tau = \frac{\ctight m_{k - 1}\rho_k^2}{4}$, where we follow the notation \eqref{eq:xsetcdef}. By the same induction argument as in Proposition~\ref{prop:sparse_plus_small}, we maintain $\normop{\alla(\1_{m_{k - 1}})} \le 2$ at all times, so combining Lemmas~\ref{lem:minflam_constraint} and~\ref{lem:linear_opt_oracle} (with $\beta \gets \rho_k$) shows that we can achieve additive error $\tau$ with high probability in $n$ in time
\[O\Par{m\log^{1 + o(1)}(m)\log\Par{ \frac 1 \eps} \cdot \Par{\frac{\log^{3.5}(m)}{\rho_k^{13.5}} + \frac{\log^2(m)}{\rho_k^{10}} \cdot \alpha^2\troute(\alpha) }},\]
for any $\alpha \ge 1$, where we bounded $\sum_{e \in E_k} \tmv(\ma_e) = O(m\log^{1 + o(1)}(m)\log\frac 1 \eps)$ as in Section~\ref{ssec:graph_bss}. 

Let $\tx_{k}$ be the resulting approximate minimizer with more coordinates $\le - 1 + \rho_k$, and note that there is a set $S_k$ with size $|S_k| \ge \frac{\ctight m_{k - 1}}{12}$ such that $[\tx_k]_e \le -1 + \rho_k$ for all $e \in S_k$. By applying Lemma~\ref{lem:degree-preserve-round} to the graph $(V, S_k, [(\1_{m_{k - 1}} + \tx_k) \circ w_{k - 1} \circ w_G]_{S_k})$, we can produce $x_k$ which equals $\tx_k$ outside $S_k$, has at least $\frac{\ctight m_{k - 1}}{30}$ coordinates equal to $-1$, maintains $\mb_k^\top([w_{k - 1} \circ x_k]_{E_k}) = \0_n$, and
\[\normop{\alla([x_k - \tx_k]_{S_k})} \le \rho_k\normop{\alla(\1_{m_k})} = O(\rho_k).\]
The above inequality used that $x_k - \tx_k$ is bounded by $\pm \rho_k$ on $S_k$ by the guarantees of Lemma~\ref{lem:degree-preserve-round}. Combining this with $\normop{\alla(\tx_k)} = O(\rho_k)$ completes the proof.
\end{proof}

\begin{remark}\label{rem:bipartite_route}
We remark that in the proof of Theorem~\ref{thm:main_degree}, if the input graph $G$ was bipartite, then there is no need to perform the step of randomly sampling a bipartition and restricting our algorithm to the edges crossing the bipartition (as Lemma~\ref{lem:bipartite_b} applies generically in this case, and bipartiteness is preserved by our reweightings). As a result, all of the graphs for which we construct oblivious routings for via the machinery in Section~\ref{ssec:linear-oracle} are $1.1$-multiplicative spectral approximations to $G$ (taking $\eps$ sufficiently small), by the same induction argument as in Proposition~\ref{prop:sparse_plus_small}.
\end{remark}

\subsection{Runtime-sparsity tradeoffs}\label{ssec:tradeoffs}

In this section, we describe how to achieve a slightly worse sparsity guarantee than Theorem~\ref{thm:main_degree} in near-linear time. Our technique in this section leverages an improved oblivious routing procedure in the case that our input graph is an expander. For completeness, we briefly define expander graphs.
\newcommand{\vol}{\mathrm{Vol}}
\begin{definition}[Expander graphs]
Let $G = (V,E,w_G)$ be a graph. For a set $S \subseteq V$, we let $\vol(S) = \sum_{ v \in S} \delta_v$ be the sum of weighted degrees in $S$ and $\partial S = \{ e \in E : e \notin S \text{ and } e \notin E / S \}$ to be the edge boundary of $S$. The cut value of $S$ is $w(\partial S) = \sum_{e \in \partial S} [w_G]_e$. The \emph{conductance} of $S$ is 
\[
\Phi(S) = \frac{ w(\partial S) }{ \min \{ \vol(S), \vol(V/S)\} }.
\]
We say $G$ is a $\phi$-expander if $\Phi(S) \geq \phi$ for all $S \subseteq V$. 
\end{definition}
The main fact we use is that electric flows induce good oblivious routings on expanders. This is summarized in the following lemma from prior work, where $\mr$ is an electric routing matrix.
\begin{lemma}[Lemma~28, \cite{KelnerLOS14}]
\label{lem:expander_route}
Let $G = (V,E,w_G)$ be a $\phi$-expander. Define 
\[
\mr = \mw \mb \lap^\dagger.
\]
Then $\mr$ is an $\alpha$-oblivious routing with $\troute(\alpha) = \O(m)$ for $\alpha = O( \frac{\log m}{\phi^2})$. 
\end{lemma}

Leveraging Lemma~\ref{lem:expander_route}, we are now ready to give a $\tO(m)$ time algorithm for constructing linear-sized degree-preserving sparsifiers of bipartite expanders.

\begin{lemma}
\label{lem:expander_sparsify}
Let $G=(V,E,w)$ be a bipartite $\phi$-expander with $m = |E|$ and $n = |V|$. Given a parameter $\eps \in (0,1)$, there is a randomized algorithm which returns $w \in \R_{\geq 0}^E$ satisfying $\nnz(w) = O(\frac{n}{\eps^2})$, $|\mb|^\top w = |\mb|^\top w_G$ in time
\[
O\Par{\frac{m}{\phi^2} \cdot \textup{poly}\Par{\frac{\log m}{\eps}}}
\]
such that with high probability in $n$,
\[(1 - \eps)\lap_G \preceq \sum_{e \in E} w_e b_e b_e^\top \preceq (1 + \eps)\lap_G.\]
\end{lemma}
\begin{proof}
We employ \Cref{thm:main_degree} on $G$, and use the routings from Lemma~\ref{lem:expander_route}. In light of \Cref{rem:bipartite_route}, all graphs we compute oblivious routings on are $1.1$-sparsifiers of $G$ and are hence $\Omega(\phi)$-expanders. As such, \Cref{lem:expander_route} yields $\troute(\alpha) = \O(m)$ with $\alpha = O( \frac{\log m}{\phi^2})$ for all graphs we construct oblivious routings for: substituting this into \Cref{thm:main_degree} yields the result. 
\end{proof}

We now upgrade this sparsification result on expander graphs to general graphs. Our main tool for this is a expander decomposition result from prior work \cite{SaranurakW19}. 

\begin{lemma}[Theorem 4.1, \cite{SaranurakW19}]
\label{lem:expander_decomp}
There is a randomized algorithm which takes as input $\phi \in (0, 1)$ and a graph $G = (V,E,w_G)$ with $m = |E|$, where $w_G$ has multiplicative range $\textup{poly}(m)$, and total weight $W \defeq \sum_{e \in E} [w_G]_e$. With high probability in $n$, it partitions $V$ into $\{V_i\}_{i \in [k]}$ in time $O(\frac m \phi \log^5 m)$ such that the following properties hold.
\begin{itemize}
    \item The induced subgraph of $G$  on each $V_i$ is a $\phi$-expander.
    \item The total edge weight cut by the $V_i$ is $O(\phi W \log^3 m)$. 
\end{itemize}
\end{lemma}

We combine this primitive with \Cref{lem:expander_sparsify} to obtain the main result of this section.

\begin{theorem}\label{thm:degree_polylog}
Given a graph $G = (V, E, w_G)$ with $m = |E|$, $n = |V|$, and $\eps \in (0, 1)$, assume $w_G$ has multiplicative range $\textup{poly}(m)$. There is a randomized algorithm which returns $w \in \R_{\geq 0}^E$ satisfying $\nnz(w) = O(\frac{n \log(m)}{\eps^2})$, $|\mb|^\top w = |\mb|^\top w_G$ in time $\O( m \cdot \eps^{-O(1)})$ such that with high probability in $n$, 
\[(1 - \eps)\lap_G \preceq \sum_{e \in E} w_e b_e b_e^\top \preceq (1 + \eps)\lap_G.\]
\end{theorem}
\begin{proof}
We begin applying \Cref{lem:expander_decomp} to find bipartite expanders covering a constant fraction of $G$'s edge weight. Let $S$ be a set which includes each vertex in $G$ with probability $\half$: note that since each edge is cut by $S$ with probability $\half$ $\E [ w(\partial S) ] = \frac{1}{2} W$ by linearity of expectation. By repeating this $O(\log n)$ times, we find a set $S$ with  $w(\partial S) \geq \frac{1}{3} W$ with high probability in $n$. Let $H$ be the graph containing all edges cut by $S$, and note that it is bipartite. Applying \Cref{lem:expander_decomp} to $H$ with $\phi = O(\log^{-3}(m))$, we obtain induced subgraphs $\{H[V_i]\}_{i \in [k]}$ of $H$ satisfying the following.
\begin{itemize}
    \item Each $H[V_i]$ is a $\phi$-expander.
    \item The $H[V_i]$ collectively contain $\frac 3 4$ of the edge-weight of $H$. 
\end{itemize}
As subgraphs of bipartite graphs are bipartite, these $H[V_i]$ are bipartite $\phi$-expanders which collectively cover $\frac{3}{4} \cdot \frac{1}{3} = \frac{1}{4}$ of the edge weight of $G$. By iterating this procedure on the remaining edges, we thus obtain disjoint subgraphs $\{G_i\}_{i \in [K]}$ of $G$ for $K = O(\log m)$ satisfying the following.
\begin{itemize}
    \item $G = \bigcup_{i \in [K]} G_i$.
    \item Each $G_i$ is a disjoint union of bipartite $O(\log^{-3}(m))$-expanders. 
\end{itemize}
We apply \Cref{lem:expander_sparsify} to each expander in each $G_i$ and combine the outputs. As the union of sparsifiers is a sparsifier of the union, the resulting graph is an $\eps$-spectral sparsifier of our original $G$. Additionally, as the total number of vertices in all of the $G_i$ is at most $nK = O(n \log m)$, the number of edges in the resulting sparsifier is bounded as $O(\frac{n \log m}{\eps^2})$. Finally, as the $G_i$ are disjoint they collectively contain $m$ edges: the total runtime of the calls to \Cref{lem:expander_sparsify} is thus the desired
\[
O\Par{m \cdot \textup{poly}\Par{\frac{\log m}{\eps}}}.
\]
\end{proof} %
\section{Spencer's theorem}\label{sec:spencer}

In this section, we describe how we recover the main result of \cite{JainSS23} within our framework. Specifically, we prove the following result which is identical to Theorem 1.1 of \cite{JainSS23} after explicitly dropping any empty rows and columns of $\ma$. We remark that for this section only, we will reverse the roles of ``$m$'' and ``$n$'' for consistency with \cite{JainSS23} (so there are $n$ colors to choose).

\begin{theorem} \label{thm:SpencerNearlyLinearTime}
Let $\ma \in \R^{m \times n}$ satisfy $\max_{i \in [m] ,j \in [n]} |\ma_{ij}| \leq 1$. There is an algorithm running in $\bO(\nnz(\ma) \cdot \log^5 (n))$ time which, with probability at least $\half$, returns $x \in \{\pm 1 \}^n$ such that, for an absolute constant $C$,
\[
\norm{\ma x}_\infty \leq C \sqrt{n \log\Par{\frac m n + 2}}.
\]
\end{theorem}

As with the \cite{JainSS23} result, the failure probability may be boosted using independent repetitions. Our strategy to prove Theorem~\ref{thm:SpencerNearlyLinearTime} is the framework in Section~\ref{sec:framework}. We need the following standard fact about the Gaussian measure of the set of Spencer partial colorings, see e.g.\ Lemma 8.9,~\cite{RothvossLectureNotes2016}.

\begin{lemma}\label{lem:spencer_measure}
Let $\ma \in \R^{m \times n}$ have $\max_{i \in [m],j \in [n]} |\ma_{ij}| \leq 1$. Then $\gamma_n(\set) \ge \exp(-\frac n {40})$ where 
\begin{equation}\label{eq:spencer_body}
\set \defeq \Brace{x \in \R^n \mid \norm{\ma x}_{\infty} \leq \sqrt{8 n \log\Par{\frac m n + 2}} }.
\end{equation}
\end{lemma}

To handle error induced by near-tight constraints, we use a randomized rounding procedure.

\begin{lemma}\label{lem:RoundSpencer}For $\set$ in \eqref{eq:spencer_body} and universal constants $c, C$, suppose we have $x \in C\set \cap [-1,1]^n$ with 
\[\min(|1-x_i|, |1+x_i|) \le \frac{1}{\sqrt{\log(m)}} \textup{ for all } i \in S,\]
where $S \subseteq [n]$ has $|S| \ge c n$. We can compute $x' \in C'\set \cap [-1,1]^n$ for a universal constant $C'$ with at least $\frac{cn}{3}$ coordinates in $\{-1,1\}$ in $O(n)$ time, with high probability in $n$.
\end{lemma}

\begin{proof}

Set $x'_i = x_i$ for $i \not \in S$ and, for $i \in S$, sample uniformly at random $x'_i \sim \{1, 2x_i-1\}$ if $x_i > 0$ and $x'_i \sim \{-1,2x_i+1\}$ if $x_i < 0$. Then $\E[x'_i] = x_i$ and by the triangle inequality, \[\norm{\sum_{i \in [n]} x'_i A_i}_{\infty} \le  \norm{\sum_{i \in S} (x'_i - x_i) A_i}_{\infty} + \norm{\sum_{i \in [n]} x_i A_i}_{\infty}.\]  
By a Chernoff bound, at least $\frac 1 3$ of the coordinates in $S$ will be rounded to $\pm 1$ with probability $1 - \exp(-\Omega(n))$, and clearly $x' \in [-1, 1]^n$. Moreover, the $x'_i - x_i$ are Rademacher random variables scaled by a factor of at most $\frac{1}{\sqrt{\log(m)}}$, so Hoeffding's inequality yields the claim $x' \in C' \set$.
\end{proof}

Finally, we give an efficient algorithm based on stochastic mirror descent for solving subproblems required by our framework. The next result follows straightforwardly from arguments in \cite{CarmonJST20} (building on \cite{ClarksonHW12}), but we summarize how to obtain it for completeness in Appendix~\ref{app:modified_boxspec}.

\begin{restatable}{lemma}{restategames}\label{lem:l2-l1_games}
Let $\ma \in \R^{m \times n}$ have $\le k$ nonzero entries in every column and $\max_{i \in [m]} \|\ma_{i:}\|_2 \le R$, and let $v \in \R^n$, $\lam, \eps \ge 0$, and $\delta \in (0, 1)$. There is an algorithm which runs in time 
\[O\Par{m + \frac{(k + n) n R^2\log^2(m)\log (\frac 1 \delta)}{\eps^2}},\]
and returns $x \in [-1, 1]^n$ such that with probability $\ge 1 - \delta$,
\[\norm{\ma x}_\infty + \lam\norm{x - v}_2^2 \le \min_{x' \in [-1, 1]^n} \norm{\ma x'}_\infty + \lam\norm{x' - v}_2^2 + \eps.\]
\end{restatable}

Now we have all the tools to prove Theorem~\ref{thm:SpencerNearlyLinearTime}.

\begin{proof}[Proof of Theorem~\ref{thm:SpencerNearlyLinearTime}]
We first state several simplifying reductions from \cite{JainSS23}. First, the proof of Theorem 1.1 shows it suffices to prove Theorem~\ref{thm:SpencerNearlyLinearTime} when $\frac n {\log^2 n} \le m \le n^2$; for small $m$, \cite{AlweissLS21} proves the result, and for large $m$, a random coloring suffices. Next, the reduction from Theorem 1.1 to Theorem 2.1 in \cite{JainSS23} uses \cite{AlweissLS21} to handle columns with less than $n \cdot \log^{-2}(n)$ nonzero entries, so after applying this reduction we may assume $n^2 = O(\nnz(\ma) \cdot \log^2(n))$. The reduction from Theorem 2.1 to Theorem 2.2 in~\cite{JainSS23} shows that it suffices to solve the partial coloring variant of the problem $O(\log\log n)$ many times, i.e.\ to produce a point in $O(1)\set \cap [-1, 1]^n$ with $\Omega(n)$ tight constraints, where $\set$ is defined as in \eqref{eq:spencer_body}. Finally, the proof of Theorem 2.2 in \cite{AlweissLS21} shows that columns with a large number of nonzero entries may be colored separately, so we can reduce to no column having more than $k = O(\nnz(\ma) \cdot \frac{\log(n)}{n})$ nonzero entries. 

We now solve the partial coloring problem under the restrictions that the maximum column sparsity is $k = O(\nnz(\ma) \cdot \frac{\log(n)}{n})$, and $n^2 = O(\nnz(\ma) \cdot \log^2(n))$.
We apply Proposition~\ref{prop:binary_search_lam} with $\Theta = n$, $\beta = \frac {1} {\sqrt{\log m}}$, and $\rho = \Cset\sqrt{8n\log(\frac m n + 2)}$, where $\Cset$ is the constant from Proposition~\ref{prop:many_tight}.
Proposition~\ref{prop:binary_search_lam} then requires solving $O(\log \log n)$ subproblems of the form 
\[\arg \min_{x \in [-1, 1]^m}  \|\ma x \|_\infty + \lam\norm{x - g}_2^2  \]
to error $\eps \defeq \Omega(\frac {\sqrt n} {\log n})$. We apply Lemma~\ref{lem:l2-l1_games} to do so with error probability $O(\log^{-1}(n))$ in time 
\[O\Par{n^2 + (k + n)n \log^3(n)\log\log(n)} = \bO\Par{\nnz(\ma) \cdot \log^5(n)}\]
where we use $R = \sqrt n$ and our assumed bounds on $k$ and $n^2$. A union bound over $\bO(1)$ subproblems (with the overhead of $\bO(1)$ partial coloring phases) gives the failure probability. The partial coloring follows by applying Lemma~\ref{lem:spencer_measure}, Proposition~\ref{prop:many_tight}, Lemma~\ref{lem:many_nearly_tight}, and Lemma~\ref{lem:sc_distsquare} to show the output of Proposition~\ref{prop:binary_search_lam} has $\Omega(n)$ nearly-tight constraints, and then using Lemma~\ref{lem:RoundSpencer} to round these nearly-tight constraints to tight ones. Finally, the runtime follows since there are $\bO(1)$ calls to Lemma~\ref{lem:l2-l1_games}.

\end{proof} 
\section*{Acknowledgements}
AJ and KT thank Aaron Sidford for several helpful conversations, and Yang P.\ Liu for the suggestion to use dynamic data structures for our rounding procedure in Section~\ref{sec:degree}. VR thanks Thomas Rothvoss for helpful discussions about the Gaussian measure lower bound in Section~\ref{ssec:reduction}.

\newpage
\bibliographystyle{alpha}
\newcommand{\etalchar}[1]{$^{#1}$}

\newpage
\begin{appendix}
\section{Proof of Theorem~\ref{thm:gen_rr20}}\label{app:modified_measure}

In this section we provide a proof of Theorem~\ref{thm:gen_rr20}. As much of the proof is the same, we defer exposition to \cite{RR20} for brevity and only provide technical details here. First we justify that we may assume $\tau \ge 1$. If the result holds for $\tau = 1$, then for any $\tau < 1$ we may consider instead the matrices $\ma'_i := \frac 1 \tau \ma_i$ which still satisfy $\sum_{i=1}^m |\ma'_i| \preceq \id_n$ and $\sum_{i=1}^n \Tr(|\ma'_i|) = 1$. Then in order to show a measure lower bound for the discrepancy body associated with $\ma_i$, it suffices to do so for the (smaller) body associated with $\ma'_i$. Hence we will assume $\tau \ge 1$ from now on. 

To parallel \cite{RR20}, we let
\begin{equation}\label{eq:msdef}\ms := \alla(\1_m),\quad \tms := \frac{1}{2} \Big(\ms + \frac{\tau}{n} \cdot \id_n\Big),\quad \ma_{C,D} \defeq (C+ D\|x\|_2^2) \cdot \id_n - \alla(x).\end{equation}
Note that $\tms \succeq \frac 1 {2n} \id_n$ by definition. The potential function we will use is:
\[
\Phi_{C,D} (x) := \Tr(\tilde{\ms} \cdot \ma_{C,D} (x)^{-1}).
\]

In order to analyze changes to the potential, we require a generalization of Lemma 11 in~\cite{RR20}.

\begin{lemma} \label{lem:MatrixTaylorApprox}
Let $\ma,\mb,\ms \in \R^{n \times n}$ be symmetric with $\ma, \ms \succ \mzero_n$ and $\normsop{\delta \ma^{-\half} \mb \ma^{-\half}} \leq \frac{1}{2}$.
Then $
 \Tr((\ma-\delta \mb)^{-1}\ms) = \Tr(\ma^{-1}\ms) + \delta \Tr(\ma^{-1}\mb\ma^{-1}\ms) + c\delta^2 \Tr(\ma^{-1}\mb\ma^{-1}\mb\ma^{-1}\ms)
$ for some $c \in [0, 2]$.
\end{lemma}

\begin{proof}
We abbreviate $\mm := \delta \ma^{-1}\mb$. As $\|\mm\|_{\textrm{op}} \leq \frac{1}{2}$, the matrix $\id_n-\mm$
is non-singular. We obtain
\[
(\ma - \delta \mb)^{-1} = (\ma (\id_n - \mm))^{-1} =(\id_n - \mm)^{-1} \ma^{-1} = \sum_{k=0}^\infty \mm^k \ma^{-1},
\]
so that 
\begin{align*} \Tr((\ma-\delta \mb)^{-1}\ms) &=  \sum_{k=0}^\infty \Tr(\mm^k \ma^{-1} \ms) =  \sum_{k=0}^2 \Tr(\mm^k \ma^{-1} \ms) + \sum_{k=3}^\infty \Tr(\mm^k \ma^{-1} \ms) \\
&= \Tr(\ma^{-1}\ms) + \delta\Tr(\ma^{-1} \mb \ma^{-1} \ms) + \delta^2 \Tr(\ma^{-1}\mb\ma^{-1}\mb\ma^{-1}\ms) \\
&+ \sum_{k=3}^\infty \Tr(\mm^k \ma^{-1} \ms).
\end{align*}
We proceed to bound the last sum above. For any $k \ge 3$,
\begin{align*}|\Tr(\mm^k \ma^{-1} \ms)| &= \Abs{\delta^{k}\inprod{\Par{\ma^{-\half}\mb\ma^{-\half}}^{k - 2} }{\ma^{-\half}\mb \ma^{-1} \ms \ma^{-1} \mb \ma^{-\half}}} \\
&\le \normop{(\delta\ma^{-\half}\mb\ma^{-\half})^{k - 2}} \cdot \delta^2 \Tr\Par{\ma^{-\half}\mb \ma^{-1} \ms \ma^{-1} \mb \ma^{-\half}} \\
&\le \normop{\delta\ma^{-\half}\mb\ma^{-\half}}^{k - 2} \cdot \delta^2 \Tr(\ma^{-1}\mb\ma^{-1}\mb\ma^{-1}\ms) \\
&\le \frac{1}{2^{k - 2}} \delta^2 \Tr(\ma^{-1}\mb\ma^{-1}\mb\ma^{-1}\ms),\end{align*}
where the first line used the cyclic property of trace, the second was the matrix H\"older inequality, and the third used $\normop{\ma^k} = \normop{\ma}^k$ for symmetric $\ma$. The conclusion then follows from combining the above two displays and applying the triangle inequality.
\end{proof}

We next give a variant of how the potential changes in a single step (Lemma 12 in \cite{RR20}).

\begin{lemma} \label{lem:PotentialFunctionUpdateInOneIteration}
Fix $\alpha \in (0, \half)$, $D \le 1$, and $m = \Omega(\alpha^{-2})$ for a sufficiently large constant.
Let $x \in \R^m$ and suppose $\ma_{C,D}(x) \succ \mzero$, $\Phi_{C,D}(x) \leq \frac{Dm^2\alpha^2}{10}$, and $\delta := \frac{\alpha}{45m^5 n^2}$.
Define  $F(y)$ so that
\begin{equation}\label{eq:Fydef}
\Phi_{C + \delta^2 F(y),D}(x+\delta y) = \Phi_{C,D}(x).
\end{equation}
Then there is a $\mx \in \R^{m \times m}$ with $\mzero \preceq \mx \preceq \id_m$ and $\Tr(\mx) \geq (1-\alpha^2) \cdot m$ so that $\E_{y \sim \Nor_{\leq m}(\0_m,\mx)}[F(y)] \le 0$
while always $|F(y)| \leq 7Dm^4 n$. Further, $\ma_{C + \delta^2 F(y), D} (x + \delta y) \succ \mzero$.
\end{lemma}

\begin{proof}[Proof of Lemma~\ref{lem:PotentialFunctionUpdateInOneIteration}]
To simplify notation, we abbreviate matrices 
\[
\ma := \ma_{C,D}(x), \quad  \tmb :=  \alla(y), \quad \mb := \tmb - \delta (D\|y\|_2^2  + F(y)) \id_n.\]
Next, we define an index set
\[
\mathcal{I} := \Big \{i \in [m] \mid \Tr(\ma^{-1} \ma_i) \leq \frac{1.1}{\alpha^2 m} \cdot \Tr(\ma^{-1} \ms)\Big \}.
\]
By Markov's inequality, we have $|\mathcal{I}| \geq (1 - \frac{\alpha^2}{1.1})m$. Consider the subspace
\begin{align*}
H :=\Big\{y \in \R^m \mid y_i = 0 \; \forall i \notin \mathcal{I}, \left<x,y\right> = 0, \inprod{\alla(y)}{\ma^{-1}\tms\ma^{-1}} = 0, & \inprod{\alla(y)}{\ma^{-1}\tms\ma^{-2}} = 0, \\ &\inprod{\alla(y)}{\ma^{-2}\tms\ma^{-1}} = 0 \Big\}
\end{align*}
so $\dim(H) \ge |\mathcal{I}| -4 \ge (1-\alpha^2)m$. 
We choose $\mx$ as the projection matrix onto $H$. The remaining proof is organized in two claims: the first bounds $|F(y)|$, and the second controls $\E F(y)$.

\begin{claim}\label{claim:boundF}
If $y \in H$ and $\norm{y}_2 \le m$, $|F(y)| \le 7Dm^4 n$, $\normop{\mb} \le 4m$, and $\normsop{\delta \ma^{-\half} \mb \ma^{-\half}} \le \half$.
\end{claim}
\begin{proof}
Since \eqref{eq:Fydef} holds, monotonicity of the inverse shows the difference matrix
\[\ma_{C+\delta^2 F(y),D}(x+\delta y)-\ma_{C,D}(x) = \delta^2 (D\|y\|_2^2 + F(y)) \cdot \id_n - \delta \alla(y)\]
has at least one positive and one negative eigenvalue. Since $\delta \alla(y) \succeq -\delta\norm{y}_\infty \id_n \succeq -\delta m \id_n$, in order for the difference matrix to have a positive eigenvalue we must have $F(y) \ge -\frac {2m} \delta$,
and similarly $\delta \alla(y) \preceq \delta m \id_n$ implies $F(y) \le \frac {m} \delta$. Hence, $|F(y)| \le \frac {2m} \delta$, which implies the bound:
\begin{equation}\label{eq:bbound}\normop{\mb} \le \normop{\alla(y)} + \delta\Par{Dm + \frac m \delta} \le 4m.\end{equation}
Further, since $\normsop{\tms \ma^{-1}} \le \Phi_{C, D}(x) \le \frac{Dm^2}{10}$, we have
\begin{equation}\label{eq:abound}
\begin{aligned}
\normsop{\delta \ma^{-\half}\mb\ma^{-\half}} &\le \delta \normsop{\ma^{-\half}}^2 \normsop{\mb} \le 4m\delta \normop{\ma^{-1}} \\
&\le 4m\delta \normsop{\tms^{-1}}\normsop{\tms \ma^{-1}} \le \frac {2Dm^3\delta} 5 \normsop{\tms^{-1}} \le \frac{4Dm^3n \delta}{5} \le \half,
\end{aligned}
\end{equation}
where we used our bounds on $\mb$ and $\tms \ma^{-1}$, and recalled from \eqref{eq:msdef} that $\tms \succeq \frac 1 {2n} \id_n$. It remains to obtain a stronger bound on $|F(y)|$. By using Lemma~\ref{lem:MatrixTaylorApprox} and the definition of the subspace $H$,
\begin{equation}\label{eq:phi_expand}
\begin{aligned}
0 &= \Phi_{C+\delta^2 F(y),D}(x + \delta y) - \Phi_{C,D}(x) = \Tr\Par{\Par{\ma - \delta \mb}^{-1} \tms} - \Tr\Par{\ma^{-1}\tms} \\
&= \delta\Tr\Par{\ma^{-1}\mb\ma^{-1}\tms} + c\delta^2 \Tr\Par{\ma^{-1}\mb\ma^{-1}\mb\ma^{-1}\tms} \\
&= \delta^2\Par{-\Par{D\norm{y}_2^2 + F(y)}\Tr(\ma^{-2}\tms) + c\Tr\Par{\ma^{-1}\mb\ma^{-1}\mb\ma^{-1}\tms}},
\end{aligned}
\end{equation}
for some $c \in [0, 2]$, since $\Tr(\ma^{-1}\alla(y)\ma^{-1}\tms) = 0$. This implies
\[|F(y)| \le \frac{c\Tr\Par{\ma^{-1}\mb\ma^{-1}\mb\ma^{-1}\tms}}{\Tr\Par{\ma^{-2}\tms}} + D\norm{y}_2^2 \le 2\normop{\ma^{-1}}\normop{\mb}^2 + Dm^2 \le 7Dm^4 n,\]
where we used the matrix H\"older inequality, \eqref{eq:bbound}, and \eqref{eq:abound} which implies $\normsop{\ma^{-1}} \le \frac{Dm^2 n}{5}$.
\end{proof}

\begin{claim}\label{claim:efbound}
We have
\[\E_{y \sim \Nor_{\le m}(\0_m, \mx)}\Brack{\Tr\Par{\ma^{-1}\mb\ma^{-1}\mb\ma^{-1}\tms}} \le \Par{\frac{ Dm }{200}+ \frac{1.1}{\alpha^2 m}\Tr\Par{\ma^{-1}\ms}}\Tr\Par{\ma^{-2}\tms}.\]
\end{claim}
\begin{proof}
Let $\lam \defeq \delta(D\norm{y}_2^2 + F(y)) \le 7.7\delta Dm^4 n$. By the last two constraints in the definition of $H$, and using that $\mb = \tmb - \lam\id_n$, for any $y \in H$ with $\norm{y}_2 \le m$,
\begin{align*}
\Abs{\Tr\Par{\ma^{-1}\mb\ma^{-1}\mb\ma^{-1}\tms} - \Tr\Par{\ma^{-1}\tmb\ma^{-1}\tmb\ma^{-1}\tms}} &= \lam^2\Tr\Par{\ma^{-3}\tms} \le \lam^2 \normop{\ma^{-1}} \Tr\Par{\ma^{-2}\tms} \\
&\le 12\delta^2D^3m^{10}n^3 \Tr\Par{\ma^{-2}\tms}.
\end{align*}
Moreover, defining $\mw_i \defeq \tms^{\half}\ma^{-1}\ma_i\ma^{-\half}$ for all $i \in [n]$,
\begin{align*}
\E_{y \sim \Nor_{\le m}(\0_m, \mx)}\Brack{\Tr\Par{\ma^{-1}\tmb\ma^{-1}\tmb\ma^{-1}\tms}} &= \E_{y \sim \Nor_{\leq m}(\0_m,\mx)}\Brack{\sum_{i \in \mathcal{I}} \sum_{j \in \mathcal{I}} y_iy_j\inprod{\mw_i}{\mw_j}} \\
&= \E_{y \sim \Nor_{\leq m}(\0_m,\mx)}\Brack{\norm{\sum_{i \in \mathcal{I}} y_i \mw_i}_{\textup{F}}^2} \le \sum_{i \in \mathcal{I}} \norm{\mw_i}_{\textup{F}}^2 \\
&= \sum_{i \in \mathcal{I}} \Tr\Par{\ma^{-1}\ma_i \ma^{-1}\ma_i \ma^{-1} \tms} \\
&\le \sum_{i \in \mathcal{I}}\Tr\Par{\ma^{-1}\ma_i\ma^{-1}\tms} \Tr\Par{\ma^{-1}\ma_i} \\
&\le \frac{1.1}{\alpha ^2 m}\Tr\Par{\ma^{-1}\ms}\Tr\Par{\ma^{-2}\tms}.
\end{align*}
The second line used Lemma 9 of \cite{RR20} with the fact that when $y$ is truncated, the norm is zero. The fourth line used Lemma 10 of \cite{RR20}, and the last used the definition of $\mathcal{I}$ and $\alla(\1_m) \preceq \id_n$. The conclusion follows by combining the above two displays with our choice of $\delta$.
\end{proof}

Taking expectations over \eqref{eq:phi_expand} and applying Claim~\ref{claim:efbound} shows
\begin{align*}
0 \le \Par{-0.44Dm - \E_{y \sim \Nor_{\leq m}(\0_m,\mx)}[F(y)] + \frac{2.2}{\alpha^2 m}\Tr\Par{\ma^{-1}\ms}}\Tr\Par{\ma^{-2}\tms},
\end{align*}
where we used $\dim(H) \ge \half m$ and Corollary 5 of \cite{RR20} to lower bound $\E_{y \sim \Nor_{\leq m}(\0_m,\mx)}[\norm{y}_2^2] \ge 0.45Dm$. Rearranging, and using $\ms \preceq 2\tms$ and our potential bound, then yields
\begin{align*}
\E_{y \sim \Nor_{\leq m}(\0_m,\mx)}[F(y)] \le \frac{4.4}{\alpha^2 m}\Tr(\ma^{-1}\tms) - 0.44Dm \le 0.
\end{align*}
Finally, note that since $\Tr(\ma^{-1}\tms) \le \frac{Dm^2\alpha^2}{10}$ and $\tms \succeq \frac 1 {2n} \id_n$, no eigenvalue of $\ma$ can be smaller than $\frac{5\tau}{Dm^2\alpha^2 n}$. Hence, the last conclusion on positive semidefiniteness follows from
\[\ma_{C + \delta^2 F(y), D}(x + \delta y) - \ma = \delta^2\Par{D\norm{y}_2^2 + F(y)}\id_n - \delta\alla(y) \succeq -3m\delta \id_n \succ -\frac{1}{2n} \id_n. \qedhere \]
\end{proof}
\restategenrr*
\begin{proof}
Assume the bounds in Lemma~\ref{lem:PotentialFunctionUpdateInOneIteration} are met, and define $\eps \defeq \sqrt{\frac \tau m}$. We overload
\[
\ma_i \gets \begin{pmatrix} \ma_i & \mzero \\ \mzero & -\ma_i \end{pmatrix} \in \R^{2n \times 2n} \text{ for all } i \in [n].
\]
We run a hypothetical algorithm for $T \defeq \frac 1 {\delta^2}$ iterations, with parameter choices
\[C \defeq \frac{10\eps}{\alpha},\; D \defeq \frac{2\eps}{\alpha m},\; \delta \defeq \frac{\alpha }{45 m^5 n^2}.\]
Our algorithm initalizes $x^{(0)} \gets \0_m$ and for $t \in [T]$, lets $\mx^{(t)}$ be the covariance matrix given by Lemma~\ref{lem:PotentialFunctionUpdateInOneIteration} applied to $x^{(t - 1)}$, samples $y^{(t)} \sim \Nor(\0_m, \mx^{(t)})$ and $z^{(t)} \sim \Nor(\0_m, \id_m - \mx^{(t)})$, and updates $x^{(t)} \gets x^{(t - 1)} + \delta y^{(t)}$. Analagously to the proof of Theorem 1 in \cite{RR20}, Azuma's inequality and Gaussian concentration show that all of the following hold with probability at least $\half$.
\begin{enumerate}
    \item $\norm{y^{(t)}}_2 \le m$ for all $t \in [T]$.
    \item $\delta^2 \sum_{t \in [T]}F(y^{(t)}) \le \frac C {10}$.
    \item $\norm{Y}_2^2 \le 5m$ and $\norm{Z}_2^2 \le 5\alpha^2 m$, where $Y \defeq \sum_{t \in [T]} y^{(t)}$ and $Z \defeq \sum_{t \in [T]} z^{(t)}$.
\end{enumerate}
We remark that the first claim above results in the requirement that $n \le 2^{\frac m 5}$, since we are union bounding over $n^4\textup{poly}(m)$ steps and Gaussian concentration fails with probability $2^{-m}$ in each step by Corollary 5 of \cite{RR20}. Further, the initial value of the potential function is $2C^{-1} \tau = \frac{Dm^2\alpha^2}{10}$, so the potential bound in Lemma~\ref{lem:PotentialFunctionUpdateInOneIteration} is met if we update $C$ by $\delta^2 F(y^{(t)})$ each iteration.  
Under these events, since $Y + Z$ is a draw from $\Nor(\0_m, \id_m)$, we have shown 
\[\normop{\alla(Y)} \le (1.1 C + D\norm{Y}_2^2) \le \frac{21\eps}{\alpha}.\]
Hence, at least half of draws from $\Nor(\0_m, \id_m)$ are within distance $\alpha\sqrt{5m}$ from $\frac{21}{\alpha}\set_{\eps, \alla}$. Reparameterizing $\alpha$ by a constant and adjusting $C_0$ then yields the claim, up to the restriction $m = \Omega(\alpha^{-2})$ in Lemma~\ref{lem:PotentialFunctionUpdateInOneIteration}. If instead $\alpha \ge \frac{C_0}{\sqrt{2m}}$, since $\normop{\alla(x)}^2 \le \normf{\alla(x)}^2 \le 2\tau$ with probability $\ge \half$ by Markov's inequality, $\frac{C_0} \alpha \set_{\eps, \alla}$ already covers the Gaussian measure without the addition of $\alpha \sqrt{m}\ball_2^m$.
\end{proof} %
\section{Optimization subroutines}\label{app:modified_boxspec}

\subsection{Discussion of Proposition~\ref{prop:modified_boxspec}}

In this section, we give a discussion of how to modify the algorithm in \cite{JambulapatiT23} to obtain Proposition~\ref{prop:modified_boxspec}. In particular, Proposition~\ref{prop:modified_boxspec} follows from modifying Theorem 3 of \cite{JambulapatiT23} applied to the regularized box-spectraplex primal-dual formulation of \eqref{eq:primaldual_opnorm}, whose notation we will follow throughout this section. We observe that $\tmv(\tma_i) = \tmv(|\tma_i|) = \tmv(\ma_i)$, since we assumed all the $\ma_i \in \PSD^d$, and there is no ``$\mb$'' term in the notation of \cite{JambulapatiT23}, Theorem 3. Further,
\[\normop{\sum_{i \in [m]} \tma_i} = \normop{\sum_{i \in [m]} \ma_i} = \normop{\alla(\1_m)},\]
so in the notation of \cite{JambulapatiT23}, Theorem 3, it suffices to set $L_{\textup{tot}} = L_{\alla} = \normop{\alla(\1_m)}$. It remains how to handle the regularization term, denoted in this section by
\[q(x) \defeq \lam\norm{x - g}_2^2.\]
It is known in the literature how to modify extragradient methods to handle composite terms (e.g.\ Section 5.2, \cite{CarmonJST19}), but because the algorithm in \cite{JambulapatiT23} is somewhat nonstandard, we give a brief description here. As seen in the proofs of Lemma 4, Corollary 2 of \cite{JambulapatiT23}, it suffices to modify the left-hand sides of Eq.\ (10) in \cite{JambulapatiT23} to read for iterates $z = (x, \my)$, $z' = (x', \my')$,
\begin{align*}
\eta\inprod{g(z) + \nabla q(x')}{z' - z^+} &\le V_z^{(\alpha + \beta)}(z^+) - V^{(\alpha)}_{z'}(z^+) - V^{(\alpha)}_z(z'), \\
\eta\inprod{g(z') + \nabla q(x')}{z^+ - u} &\le 2V_z^{(\alpha + \beta)}(u) - 2V^{(\alpha + \beta)}_{z^+}(u) - 2V^{(\alpha + \beta)}_z(z^+) \\
&+ 2V^{\gamma h}_{\bmy}(u^{\mathsf{y}}) - 2V^{\gamma h}_{\bmy^+}(u^{\mathsf{y}}),
\end{align*}
where $u^{\mathsf{y}}$ is the spectraplex component of $u \in [-1, 1]^m \times \Delta^{2n \times 2n}$, and $\bmy$, $\bmy^+$ are auxiliary iterates maintained by the extragradient step oracles of \cite{JambulapatiT23}. Achieving the second equation above is immediate by giving the extragradient step oracle (Algorithm 5) in \cite{JambulapatiT23} the linear term $g(z') + \nabla q(x')$, instead of just $g(z')$. To achieve the first equation, we modify the subproblems solved in Lines 3 and 5 of the gradient step oracle (Algorithm 4) in \cite{JambulapatiT23} to minimize a linear term, plus $r(\cdot, \my)$, plus $\eta q$. It is simple to check that first-order optimality conditions on $x'$ in Corollary 3 then result in the additional $\nabla q(x')$ term in the first line above, and no other proofs are changed.

\subsection{Proof of Lemma~\ref{lem:l2-l1_games}}

\restategames*
\begin{proof}
It suffices to prove that in time $O((k + n)n^2\log^2(m)\eps^{-2})$, we can return a point satisfying
\[\E\Brack{\norm{\ma x}_\infty + \lam\norm{x - v}_2^2} \le \min_{x' \in [-1, 1]^n} \norm{\ma x'}_\infty + \lam\norm{x' - v}_2^2 + \frac \eps 2,\]
since Markov's inequality results in an $\eps$-suboptimal point with probability $\ge \half$, and then we can run $O(\log \frac 1 \delta)$ independent runs and take the best function value. To do this, we reparameterize $\ma$ as the vertical concatenation of $\sqrt{n}\ma$ and $-\sqrt{n}\ma$, and recast the objective as the minimax problem
\begin{equation}\label{eq:minimax_l1l2}\min_{x \in [-\frac 1 {\sqrt n}, \frac 1 {\sqrt n}]^n} \max_{y \in \Delta^m} y^\top \ma x + \lam\norm{\sqrt n x - v}_2^2.\end{equation}
It then suffices to apply Proposition 2 of \cite{CarmonJST20}, following the $\ell_2$-$\ell_1$ local norm setup in Table 6, with the composite extension in Lemma 13 (setting $Q(x, y) = \lam\norm{x - v}_2^2$). Specifically, Proposition 2 of \cite{CarmonJST20} requires a local gradient estimator $\tg$ (see Definition 3 of that paper) for the operator
\[g(x, y) = \Par{\ma^\top y, -\ma x}\]
corresponding to the bilinear component of our minimax problem \eqref{eq:minimax_l1l2}. We use
\[\tg(x, y) = \Par{\ma_{i:}, -\ma_{:j} \cdot \frac{x_j}{p_j}},\]
where $i \in [m]$ is randomly selected with probability $y_i$, and $j \in [n]$ is randomly selected with probability $p_j = x_j^2\norm{x}_2^{-2}$. This estimator $\tg$ is unbiased for $g$, meeting the first criterion of Definition 3 in \cite{CarmonJST20}. Further, since all rows of $\ma$ have $\ell_2$ norm bounded by $R\sqrt n$ (after rescaling by $\sqrt n$), Lemma 21 of \cite{CarmonJST20} with $w_0$ set to the all-zeroes vector shows $\tg$ satisfies the second criterion of Definition 3 in \cite{CarmonJST20} with $L = O(R \sqrt n)$. Finally, for the $\ell_2$-$\ell_1$ local norm setup, Table 6 in \cite{CarmonJST20} shows $\Theta = O(\log m)$. In conclusion, Proposition 2 of \cite{CarmonJST20} shows that after
\[O\Par{\frac{L^2 \Theta}{\eps^2}} = O\Par{\frac{nR^2 \log(m)}{\eps^2}}\]
iterations, we obtain expected suboptimality gap $\frac \eps 2$ as desired. In each iteration, we can explicitly update the vector $x$ and recompute the sampling distribution proportional to $x^2$ in $O(n)$ time. Because the update to $y$ is $k$-sparse, Section 5.1 of \cite{CarmonJST20} (see also Lemma 5.4 of \cite{JainSS23} for a brief summary) shows how to update $y$ and maintain a sampling distribution proportional to it in time $O(k \log(m))$ per iteration. Hence, each iteration takes time $O(n + k\log(m))$ as desired.
\end{proof} %
\section{Low-distortion subgraph construction}\label{app:subgraph_construct}

In this section, we give a proof of the low-distortion subgraph construction claimed in Proposition~\ref{prop:js21_subgraph}.

\restatejssubgraph*
\begin{proof}
This construction is implicit in the proof of  Theorem~A.4 in \cite{JambulapatiS21}: we summarize the details here for completeness. Theorem~1.9 of \cite{JambulapatiS21} gives a ``path sparsification'' algorithm, which takes as input an $n$-node $m$-edge graph $G$ and parameter $q \geq 1$ and returns a $(q, O(\log^5 n))$-path sparsifier with $O(n q \log^3 n)$ edges in $O(m + nq \log^{13} n)$ time. 

In addition, Theorem~2.5 of \cite{JambulapatiS21} provides an algorithm which takes as input a graph $G$, parameters $k, \gamma \geq 2$, and a path sparsification algorithm which runs on $n’$-vertex $m’$-edge graphs and returns a $(10 \beta, \beta)$-path sparsifier (for some value $\beta$) with $S(m’,n’)$ edges in $T(m’,n’)$ time. It outputs a subgraph $H$ with at most 
\[
n-1 + O \left( \frac{m}{\gamma} + \frac{m \log \gamma}{k^2} \right) + S\left( O(m), O\Par{\frac m k} \right)
\]
edges such that $\Tr (\lap_H^{\dagger}\lap_G)\leq m \hat{\alpha}$ for 
\[
\hat{\alpha} = O \left( \exp \left( \sqrt{ 8 \log \gamma \cdot \log \left( 48 \log k \sqrt{\log \gamma} \right) } \right) \log k \sqrt{\log \gamma} \right).
\]
Additionally, the algorithm runs in $O(m + T(O(m), O(\frac m k) )$ time.

Choosing $q = O(\log^5 n')$ in Theorem~1.9 gives an algorithm which runs on $n'$-vertex $m'$-edge graphs, returns a $(10 \beta, \beta)$-path sparsifier for $\beta =  O(\log^5 n')$, and has $S(m',n') = O( n' \log^8 n')$ and $T(m',n') = O(m' + n' \log^{18} n')$. Choosing $k = \gamma \log^{18} m$ in Theorem~2.5 and combining the above gives an algorithm which takes in $G$ and outputs a subgraph $H$ in time 
\[
O\Par{m + T\Par{ O(m), O\Par{\frac m k} } } = O\Par{m + \frac m k\log^{18} m } = O(m)
\]
with 
\begin{align*}
n - 1 + O\left( \frac{m}{\gamma} + \frac{m \log \gamma}{k^2} \right) + S\Par{O(m), O\Par{\frac m k}} &= n-1 + O\left( \frac{m}{\gamma} \right) + \frac{m}{k} \log^8 m \\
&=  n-1 + O\left( \frac{m}{\gamma} \right)
\end{align*}
edges, such that  $\Tr (\lap_H^\dagger \lap_G ) \leq m \hat{\alpha}$ with 
\begin{equation}
\label{eq:alphabound}
\hat{\alpha} = O \left( \exp \left( \sqrt{ 8 \log \gamma \cdot \log \left( 48 \log \left( \gamma \log^{18} n \right) \sqrt{\log \gamma} \right) } \right) \log \left( \gamma \log^{18} n \right) \sqrt{\log \gamma} \right). 
\end{equation}

To simplify the above expression, note that
\[
\log \left(\gamma \log^{18} n \right) \sqrt{\log \gamma} \leq 18\sqrt{\log \gamma} \log \log n  + (\log \gamma)^{\frac 3 2} \leq 20 (\log \log n)^{\frac 3 2} + 20 (\log \gamma)^{\frac 3 2}.
\]
Since $\log(a+b) \leq 2\max \{ \log a , \log b\}$ for $a,b \geq 2$, we have
\begin{align*}
8 \log \gamma \cdot \log \left( 48 \log \left( \gamma \log^{18} n \right) \sqrt{\log \gamma} \right) &\leq 16 \log \gamma \cdot \Par{\log 1000 + \frac 3 2\max \Par{\log \log \gamma, \log \log \log n  } } \\
&\leq 400 + 36 \max \Par{ \log \gamma \cdot \log \log \gamma, \log \gamma \cdot \log \log \log n } \\
&\leq 400 + 36 \max \Par{ \log \gamma \cdot \log \log \gamma , \log \log n \cdot \log \log \log n }
\end{align*}
where the last inequality holds since $\log \log \log n \geq \log \log \gamma$ only when $\gamma \leq \log n$.
Substituting the above bounds into \eqref{eq:alphabound} and using $\sqrt{a+b} \leq \sqrt{a}+\sqrt{b}$, we have
\begin{align*}
\hat{\alpha} &= O\Par{ \exp \Par{ 20 + 6 \max \Par{ \sqrt{\log \gamma \cdot \log \log \gamma}, \sqrt{\log \log n \cdot \log \log \log n }} }} \Par{(\log \log n)^{\frac 3 2} + (\log \gamma)^{\frac 3 2} )} \\
&= \bO\Par{ \exp(6 \sqrt{\log \gamma \cdot \log \log \gamma}) \exp(6 \sqrt{\log \log n \cdot \log \log \log n }) (\log \gamma)^{\frac 3 2} } = O ( \gamma^{o(1)} (\log n)^{o(1)} ).
\end{align*}
Thus the algorithm described above proves the desired claim.
\end{proof}
 \end{appendix}

\end{document}